\documentclass[10pt,journal,twocolumn]{IEEEtran}

\usepackage{wrapfig,booktabs,fancyhdr,amsmath,amsfonts,tabularx,numprint}
\usepackage{cite,bm,bbm,amssymb,amsthm,url,multirow,times,enumitem,comment}
\usepackage{mathtools,siunitx,balance,tikz,adjustbox,graphicx,array}
\usepackage{algorithm}
\usepackage{algorithmic}
\usepackage{amsthm}
\usepackage[font=footnotesize]{subcaption}
\usepackage[colorlinks=true, linkcolor=black, citecolor=blue, urlcolor=blue]{hyperref}
\usepackage[switch,pagewise]{lineno}
\usepackage{acro}
\usepackage{eqparbox}

\newcommand{\sfrac}[2]{\textstyle \frac{#1}{#2}}
\newcommand{\ba}{\begin{array}}
\newcommand{\ea}{\end{array}}

\renewcommand{\Re}{\mathfrak{R}}

\DeclareMathAlphabet{\mathpzc}{OT1}{pzc}{m}{it}

\newcommand{\fc}{\ensuremath{f_\text{c}}}
\newcommand{\Ta}{\ensuremath{T_{\text{a}}}}
\newcommand{\Na}{\ensuremath{N_{\text{a}}}}
\newcommand{\Ts}{\ensuremath{T_{\text{s}}}}
\newcommand{\Pmin}{\ensuremath{P_{\text{min}}(z,\bm{\rho})}}
\newcommand{\PminC}{\ensuremath{P_{\text{min,C}}(z,\bm{\rho})}}
\newcommand{\PminN}{\ensuremath{P_{\text{min,N}}(z,\bm{\rho})}}

\newcommand{\PminCR}{\ensuremath{P_{\text{min,C}}(z,f(\tilde{\bm{\rho}}))}}
\newcommand{\PminNR}{\ensuremath{P_{\text{min,N}}(z,f(\tilde{\bm{\rho}}))}}
\newcommand{\ACf}{\ensuremath{A(z,\bm{\rho})}}
\newcommand{\ACfAst}{\ensuremath{A^{\ast}(z,\bm{\rho})}}

\newcommand{\ACc}{\ensuremath{A_{\text{C}}(z,\bm{\rho})}}
\newcommand{\ACn}{\ensuremath{A_{\text{N}}(z,\bm{\rho})}}
\newcommand{\ACcR}{\ensuremath{A_{\text{C}}(z,f(\tilde{\bm{\rho}}))}}
\newcommand{\ACnR}{\ensuremath{A_{\text{N}}(z,f(\tilde{\bm{\rho}}))}}
\newcommand{\va}{\ensuremath{v_{\text{A}}(z,\bm{x})}}
\newcommand{\ACtilde}{\ensuremath{\tilde{A}(z,\bm{x})}}
\newcommand{\ZZBRho}{\ensuremath{\sigma_{\text{ZZB}}^2(\bm{\rho})}}
\newcommand{\ZZBRhoMaxL}{\ensuremath{\sigma_{\text{ZZB}}^2(\max_L(\bm{\rho}))}}
\newcommand{\ZZBRhoR}{\ensuremath{\sigma_{\text{ZZB}}^2(f(\tilde{\bm{\rho}}))}}
\newcommand{\ZZBcRho}{\ensuremath{\sigma_{\text{ZZB,C}}^2(\bm{\rho})}}
\newcommand{\ZZBcRhoR}{\ensuremath{\sigma_{\text{ZZB,C}}^2(f(\tilde{\bm{\rho}}))}}
\newcommand{\ZZBnRho}{\ensuremath{\sigma_{\text{ZZB,N}}^2(\bm{\rho})}}
\newcommand{\ZZBnRhoR}{\ensuremath{\sigma_{\text{ZZB,N}}^2(f(\tilde{\bm{\rho}}))}}
\newcommand{\Kz}{\ensuremath{\mathcal{K}_{0}}}
\newcommand{\Ko}{\ensuremath{\mathcal{K}_{1}}}
\newcommand{\KzC}{\ensuremath{\check{\mathcal{K}}_{0}}}
\newcommand{\KoC}{\ensuremath{\check{\mathcal{K}}_{1}}}
\newcommand{\kbranch}{\ensuremath{k_{\text{b}}}}

\DeclareAcronym{awgn}{short = AWGN, long = additive white Gaussian noise}
\DeclareAcronym{ici}{short = ICI, long = intercarrier interference}
\DeclareAcronym{sinr}{short = SINR, long = signal-to-interference-and-noise ratio}
\DeclareAcronym{snr}{short = SNR, long = signal-to-noise ratio}
\DeclareAcronym{ofdm}{short = OFDM, long = orthogonal frequency-division multiplexing}
\DeclareAcronym{mimo}{short = MIMO, long = multiple-input multiple-output}
\DeclareAcronym{cfo}{short = CFO, long =  carrier frequency offset}
\DeclareAcronym{cpe}{short = CPE, long = common phase error}
\DeclareAcronym{lte}{short = LTE, long = long term evolution}
\DeclareAcronym{nr}{short = NR, long = new radio}
\DeclareAcronym{crlb}{short = CRLB, long = Cramer-Rao lower bound}
\DeclareAcronym{acf}{short = ACF, long = autocorrelation function}
\DeclareAcronym{zzb}{short = ZZB, long = Ziv-Zakai bound}
\DeclareAcronym{mmse}{short = MMSE, long = minimum mean square error}
\DeclareAcronym{lmmse}{short = LMMSE, long = linear minimum mean square error}
\DeclareAcronym{rmse}{short = RMSE, long = root mean square error}
\DeclareAcronym{mrc}{short = MRC, long = maximum-ratio combining}
\DeclareAcronym{toa}{short = TOA, long = time-of-arrival}
\DeclareAcronym{pdf}{short = PDF, long = probability density function}
\DeclareAcronym{cdf}{short = CDF, long = cumulative distribution function}
\DeclareAcronym{dft}{short = DFT, long = discrete Fourier transform}
\DeclareAcronym{jcas}{short = JCAS, long = joint communications and sensing}
\DeclareAcronym{isac}{short = ISAC, long = integrated sensing and communications}
\DeclareAcronym{prs}{short = PRS, long = positioning reference signal}
\DeclareAcronym{sdr}{short = SDR, long = software-defined radio}
\DeclareAcronym{gnss}{short = GNSS, long = global navigation satellite system}
\DeclareAcronym{tdoa}{short = TDOA, long = time-difference-of-arrival}
\DeclareAcronym{rtt}{short = RTT, long = round-trip-time}
\DeclareAcronym{aoa}{short = AOA, long = angle-of-arrival}
\DeclareAcronym{aod}{short = AOD, long = angle-of-departure}

\newtheorem{theorem}{Theorem}

\begin{document}
\title{Ziv-Zakai-Optimal OFDM Resource Allocation for Time-of-Arrival Estimation}

\author{
  \IEEEauthorblockN{Andrew M. Graff\IEEEauthorrefmark{1}, Todd E. Humphreys\IEEEauthorrefmark{2}\\
  \IEEEauthorblockA{\IEEEauthorrefmark{1}\textit{Department of Electrical and Computer Engineering, The University of Texas at Austin}} \\
  \IEEEauthorblockA{\IEEEauthorrefmark{2}\textit{Department of Aerospace
  		Engineering and Engineering Mechanics, The University of Texas at Austin}} }
}

\maketitle

\begin{abstract}
	This paper presents methods of optimizing the placement and power allocations of pilots in an \ac{ofdm} signal to minimize \ac{toa} estimation errors under power and resource allocation constraints. \ac{toa} errors in this optimization are quantified through the \ac{zzb}, which captures error thresholding effects caused by sidelobes in the signal's \ac{acf} which are not captured by the \acl{crlb}. This paper is the first to solve for these \ac{zzb}-optimal allocations in the context of \ac{ofdm} signals, under integer resource allocation constraints, and under both coherent and noncoherent reception. Under convex constraints, the optimization of the \ac{zzb} is proven to be convex; under integer constraints, the optimization is lower bounded by a convex relaxation and a branch-and-bound algorithm is proposed for efficiently allocating pilot resources. These allocations are evaluated by their \acp{zzb} and \acp{acf}, compared against a typical uniform allocation, and deployed on a \acl{sdr} \ac{toa} measurement platform to demonstrate their applicability in real-world systems.
\end{abstract}

\begin{IEEEkeywords} 
OFDM; positioning; Ziv-Zakai; convex optimization.
\end{IEEEkeywords}

\newif\ifpreprint
\preprinttrue

\ifpreprint

\pagestyle{plain}
\thispagestyle{fancy}  
\fancyhf{} 
\renewcommand{\headrulewidth}{0pt}
\rfoot{\footnotesize \bf This work has been submitted to the IEEE for possible publication. Copyright may\\be transferred without notice, after which this version may no longer be accessible.}
\lfoot{\footnotesize \bf
  Copyright \copyright~2024 by Andrew M. Graff \\ and Todd E. Humphreys}

\else

\thispagestyle{empty}
\pagestyle{empty}

\fi


\section{Introduction}
\acresetall
Demand for accurate positioning services is growing rapidly. As next-generation networks continue to trend toward denser deployments and wider bandwidths, they have garnered increased interest as a provider of positioning services, both as a precise alternative to \acl{gnss} services and as an enabler of key environment-aware features in future 6G networks. Existing networks operate almost ubiquitously through \ac{ofdm}, permitting positioning signals to be multiplexed with communications data across both time and frequency resources. These positioning resources, henceforth referred to as pilot resources, are known by users within the network, allowing receivers to correlate against the pilot resources to obtain \ac{toa} estimates for use in several positioning protocols. Wireless networks face the difficult task of balancing the allocation of resources between these two services, with positioning users demanding more precise localization, and communications users demanding faster data rates. In such dual-functional networks, positioning signals that achieve optimal localization performance while subject to power and resource constraints are desirable. This paper explores the optimal allocation of pilot resources in an \ac{ofdm} signal to minimize \ac{toa} errors subject to these constraints.

\Ac{toa} estimation at the receiver is the first step in several positioning protocols including pseudorange multilateration, \acl{tdoa}, or \acl{rtt}, all of which have seen widespread use in existing cellular networks \cite{shamaei2018exploiting,dwivedi2021positioning}. To obtain a \ac{toa} estimate, the receiver may correlate delayed copies of the known pilot signal against the received signal, creating a correlation function. The receiver may then estimate the \ac{toa} relative to its local clock as the delay corresponding to the peak power of this correlation function. In an \ac{ofdm} signal, the presence of the cyclic prefix allows this time-delay correlation to be computed as a frequency domain correlation that is orthogonal to subcarriers not modulated with pilot resources \cite{berger2010signal}, assuming sufficiently small Doppler. As a result, the \ac{acf} of the \ac{ofdm} pilot signal depends only on the placement and power allocations of the pilot resources and not the phases of the pilot resources.

The shape of the \ac{acf} is of particular interest because it regulates how \ac{toa} estimation errors change with the \ac{snr}. The \ac{acf} is typically described by its mainlobe, which is the correlation power near the true \ac{toa}; its sidelobes, which are spikes in correlation power occurring outside of the mainlobe; and grating lobes, which are images of the mainlobe caused by aliasing within the \ac{acf}. At high \ac{snr}, estimation errors tend to be concentrated near the peak of the mainlobe. As \ac{snr} decreases, however, it becomes increasingly likely that errors occur on sidelobes, which causes a sudden and drastic increase in estimation error variance due to their distance from the mainlobe. This phenomenon is referred to as the ``thresholding effect,'' which has been widely studied in \ac{toa} estimation \cite{zeira1994realizable,nanzer2016bandpass,sahinoglu2008ultra} and motivated the use of the Barankin bound \cite{barankin1949,mcaulay1971barankin} and the \ac{zzb} \cite{Ziv1969} as alternatives to the \ac{crlb}, which does not capture this effect. Large errors may also be introduced by grating lobes. If the duration of delays over which the receiver correlates is greater than the unambiguous region of the \ac{acf}, aliasing will occur and grating lobe ambiguities may appear in the correlation function that are indistinguishable from the mainlobe. Since the grating lobes are equally-powered images of the mainlobe, these errors occur irrespective of the \ac{snr}.

A pilot resource allocation designed for \ac{toa} estimation must account for all of these sources of error. Existing techniques typically transmit positioning signals that have power uniformly distributed over the spectrum, such as the \ac{prs} in LTE and 5G NR. While these power allocations create sufficiently low sidelobes in the \ac{acf}, alternative power allocations may achieve lower \ac{toa} error variances when restricted to the same power budget and signal bandwidth. The optimal allocation to minimize estimation error variance must balance the sharpness of the mainlobe against the sidelobe level and ensure that grating lobes are not present over the range of possible correlation delays. With \textit{a priori} knowledge of both the receiver's \ac{snr} and the distribution of the \ac{toa} dictating the range of correlation delays, optimal allocations could significantly reduce \ac{toa} estimation errors without increasing the power budget or signal bandwidth. These allocations may also be designed to use a limited number of subcarriers, allowing networks to keep unused resources available for communications purposes. Furthermore, the receiver may perform either coherent correlation if the carrier phase is known accurately or noncoherent correlation otherwise, causing changes to the shape of the \ac{acf} and the distribution of post-correlation noise, thereby affecting the optimal allocation.

This paper optimizes the allocation of \ac{ofdm} pilots to minimize the error variance of \ac{toa} estimates obtained from the \ac{ofdm} signal, quantified through the \ac{zzb}. The optimization of the \ac{zzb} for \ac{toa} estimation has been discussed in prior work but never solved for the important class of \ac{ofdm} signals, restricted to integer resource constraints, nor generalized to both coherent and noncoherent reception.

The \ac{zzb} is first expressed for \ac{ofdm} signals under both coherent and noncoherent reception. In both the coherent and noncoherent cases, the minimization of the \ac{zzb} with respect to pilot power allocations while subject to a total power constraint is then proven to be a convex problem with readily calculated gradients, permitting efficient selection of optimal allocations. This minimization is then further constrained to only allocate power in a fixed number of subcarriers, yielding an integer-constrained problem that is lower-bounded by a convex relaxation. A branch-and-bound algorithm is then proposed, allowing approximately optimal allocations to be computed efficiently. Optimal (and approximately optimal) allocations are found over a range of \acp{snr} for both the coherent and noncoherent versions of the convex-constrained and the integer-constrained problems. These allocations are then analyzed through their \ac{toa} error bounds and \acp{acf}. Finally, to demonstrate their real-world applicability, \ac{toa} errors are measured on an \ac{sdr} ranging platform transmitting \ac{ofdm} signals generated from these allocations.

\subsection{Prior Work}

Prior work has analyzed signals for \ac{toa} estimation with theoretical bounds such as the \ac{zzb} and \ac{crlb}. The \ac{crlb} is derived to analyze the performance of \ac{toa} estimation with \ac{ofdm} signals in \cite{PeralRosado2018} and \cite{xu2016maximum}. Meanwhile, the \ac{zzb} on \ac{toa} estimation is derived under a tapped delay line channel model with statistical channel knowledge at the receiver for frequency-hopping signals \cite{liu2010ziv} and for ultrawide bandwidth signals \cite{dardari2009ziv}. In the context of localization, the authors of \cite{gusi2018ziv} derive the \ac{zzb} on the positioning error of direct position estimation and two-step estimation based on \ac{toa} measurements. The \ac{zzb} on \ac{toa} estimation is also derived in the context of compressed sensing radar in \cite{zhang2022deterministic}. While these studies provide insights into \ac{toa} estimation errors from these bounds, they do not consider the optimal design of the signals to minimize \ac{toa} error.


Other prior work has used the \ac{crlb}, \ac{zzb}, and other criteria to optimize signal design for \ac{toa} estimation.
%
The \ac{zzb} on \ac{toa} has been evaluated for sets of parametric candidate \ac{ofdm} power allocations in \cite{graff2024purposeful,laas2021ziv,dammann2016optimizing,staudinger2017optimized,Driusso2015}, showing how allocating power toward the extremities of the band increases \ac{toa} precision in the high-\ac{snr} regime at the expense of high sidelobe levels and how power can be intelligently allocated to balance this tradeoff. However, none of these papers determine optimal power allocations and are restricted to the coherent \ac{zzb}, limiting applicability in practical systems where the carrier phase is unknown.
Similarly, Wirsing et al. evaluate how different $\sfrac{\pi}{4}\text{-QPSK}$ signaling patterns affect the coherent \ac{zzb} \cite{wirsing2020designing}.

The \ac{toa} \ac{zzb} is optimized in \cite{xiong2023snr} to create \ac{snr}-adaptive waveforms subject to a power and bandwidth constraint, proving convexity of the phase-coherent \ac{zzb} with respect to a signal's \ac{acf} and demonstrating that notable reductions in \ac{toa} error are possible through Ziv-Zakai optimization. The optimization, however, does not consider the noncoherent \ac{zzb}, restrictions to \ac{ofdm} signals, or the more complex resource constraints addressed in the current paper.
Sun et al. \cite{sun2023trade} optimize the balance of power between positioning and communications signals to maximize data rate subject to a \ac{toa} error constraint using the \ac{zzb}, but only optimize for the total signal powers and do not optimize the allocation of power across frequency.

%
The design of a multiband \ac{ofdm} signal is addressed in \cite{dun2021design}, where a sparse selection of bands is optimized subject to a \ac{toa} error variance constraint quantified by the \ac{crlb} with results verified on an \ac{sdr} testbed.
Multi-carrier power allocations are optimized using the \ac{crlb} to minimize \ac{toa} error  for signal's experiencing interference using the \ac{crlb} in \cite{karisan2011range} and to minimize both \ac{toa} error and channel estimation error in \cite{Larsen2011,montalban2013power}.
Li et al. optimize subcarrier power allocations in an \ac{ofdm} \ac{isac} system to minimize the \ac{crlb} on ranging \ac{toa} estimation subject to power and minimum capacity constraints \cite{li2023joint}. While the optimized signals in \cite{dun2021design,karisan2011range,Larsen2011,montalban2013power,li2023joint} may be effective in high \ac{snr} conditions, these designs are inapt at low \acp{snr} due to their use of the \ac{crlb}.
%
A weighted integrated sidelobe level optimization is conducted in \cite{jamalabdollahi2017high} to design signals with desirable autocorrelation properties for \ac{toa} estimation in frequency-selective fading channels.

Outside of \ac{toa} estimation, the \ac{zzb} has also seen use in the context of \ac{mimo} radar. The noncoherent \ac{zzb} was used in \cite{gupta2019design} as a metric to evaluate sparse array designs for \ac{aoa} estimation. The authors of \cite{chiriac2015ziv} also derived the \ac{zzb} for the joint estimation of the position and velocity of targets in a \ac{mimo} radar system.

Beyond signal design for \ac{toa} estimation and positioning, a large body of prior work optimizes signal design for \ac{isac} \cite{liu2022integrated} and \ac{jcas} systems that focus on radar as the sensing component.
%
Many of these \ac{isac} studies have specifically considered the optimization of \ac{ofdm} waveforms and power allocations.
The studies in \cite{liu2017adaptive,zhang2019mutual,wei2023waveform} consider \ac{ofdm} \ac{isac} systems and optimize the transmitted signal using objective functions based on the sensing mutual information.
The authors of \cite{liyanaarachchi2020joint} optimize the symbols of an \ac{ofdm} \ac{isac} system on resources allocated for radar sensing to minimize the \ac{crlb} on target \ac{toa} and Doppler estimation.
Similarly, Wang et al. optimize subcarrier resource allocations and powers in a multicarrier \ac{isac} system to maximize capacity while enforcing a minimum \ac{sinr} constraint for radar sensing, proposing a branch-and-bound algorithm for solving the allocation problem \cite{wang2019power}.
Ni et al. optimize the precoders of an \ac{ofdm} \ac{jcas} system based on the \ac{sinr}, radar mutual information, and radar \acp{crlb} on delay, gain, and \acl{aod} \cite{ni2021multi}.
Considering only an \ac{ofdm} radar system, Sen et al. propose a multiobjective optimization to design subcarrier weights in \ac{ofdm} radar waveforms to balance both sparse estimation error and detection performance \cite{sen2010multiobjective}.
This work on the optimization of \ac{ofdm} \ac{isac} and radar systems motivates the importance of optimized \ac{ofdm} signal design. However, this work does not address the optimization of \ac{ofdm} signals to minimize the \ac{zzb} on \ac{toa} estimation.



\subsection{Contributions}
The main contributions of this paper are as follows:
\begin{itemize}
	\item Expressions for the \ac{zzb} on \ac{toa} estimation error variance for \ac{ofdm} signals under both coherent and noncoherent reception.
	\item Proofs of the convexity of both the coherent and noncoherent \acp{zzb} with respect to subcarrier power allocations.
	\item A branch-and-bound algorithm for efficiently albeit approximately solving the \ac{zzb} minimization problem under additional integer resource constraints on pilot resource allocation.
	\item Analysis of the optimal (and approximately optimal) pilot allocations, \acp{acf}, and \ac{toa} error variances obtained from both the convex-constrained and integer-constrained problems under both coherent and noncoherent reception.
	\item Evaluation of the \ac{toa} error variance measured on an \ac{sdr} platform transmitting \ac{ofdm} signals using the optimized pilot allocations in comparison against an \ac{ofdm} signal with a uniform pilot power allocation.
\end{itemize}

The remainder of the paper is organized as follows. Section~\ref{section:sig} introduces the signal model, maximum-likelihood \ac{toa} estimators, \acp{acf}, and error bounds on \ac{toa} estimation. Section~\ref{section:optim} introduces the optimization problems. Section~\ref{sec:convex_constrained} frames the optimization under convex constraints, proves convexity, and derives expressions for gradients. Section~\ref{sec:integer_constrained} frames the optimization under integer constraints and proposes a branch-and-bound algorithm to efficiently solve for approximately-optimal allocations. Section~\ref{sec:num_results} analyzes the optimized allocations, \acp{acf}, and \acp{zzb}, comparing the results against a uniform allocation. Section~\ref{sec:exp_results} analyzes the performance of the optimized allocations on an \ac{sdr} platform. Finally, Sec.~\ref{sec:conclusion} ends the paper with conclusions drawn from the results.

\textbf{Notation:} Column vectors are denoted with lowercase bold, e.g., $\bm{x}$. Matrices are denoted with uppercase bold, e.g., $\bm{X}$. Scalars are denoted without bold, e.g., $x$. The $i$th entry of a vector $\bm{x}$ is denoted $x[i]$ or in shorthand as $x_i$. The Euclidean norm is denoted $||\bm{x}||$. Real transpose is represented by the superscript $T$ and conjugate transpose by the superscript $H$. The Q-function is denoted as $Q(\cdot)$. The Marcum-Q function of order $\nu$ is denoted as $Q_{\nu}(\cdot)$. The modified Bessel function of the first kind is denoted as $I_{n}(\cdot)$. Zero-based indexing is used throughout the paper; e.g., $x[0]$ refers to the first element of $\bm{x}$.


\section{Signal Model}
\label{section:sig}

Consider an \ac{ofdm} signal with $K$ subcarriers, a subcarrier spacing of $\Delta_{\text{f}}\;\SI{}{\hertz}$, and a payload $x[k]$ for subcarrier indices $k \in \mathcal{K}$, where $\mathcal{K} = \{0,1,\ldots, K-1\}$. Let $d[k]$ be the mapping from subcarrier indices to physical distances in frequency from the carrier in units of subcarriers. This map is defined as $d[k] = k$ for $k = 0,1,\ldots\sfrac{K}{2}-1$ and $d[k] = k-K$ for $k = \sfrac{K}{2}, \sfrac{K}{2}+1,\ldots, K-1$. This signal propagates through an \ac{awgn} channel at a carrier frequency $\fc$ and experiences a gain in signal power $g$, time delay $\tau_0$, phase shift $\phi_0$, and \ac{awgn} $v[k] \sim \mathcal{CN}(0,\sigma^2)$. The baseband received signal $y[k]$ is modeled in the frequency domain as
\begin{align}
	y[k] &= \alpha[k] x[k] + v[k],\\
	\alpha[k] &= \sqrt{g}\exp\left(-j2\pi d[k] \Delta_{\text{f}} \tau_0 + j\phi_0\right).
\end{align}
This model assumes that the receiver begins sampling the transmitted \ac{ofdm} symbol during the cyclic prefix, permitting the receiver to process the signal in the frequency domain.

Assuming that the payload $x[k]$ is fully known to the receiver, the receiver can obtain a maximum-likelihood estimate of $\tau_0$ from the frequency domain correlation of $y[k]$ and $x[k]$. Let $\tau$ be the time delay at which the correlation function is evaluated. Through the remainder of this paper, this time delay will be normalized by the \ac{ofdm} sampling period $\Ts = \sfrac{1}{K \Delta_{\text{f}}}$ to define $z \triangleq \sfrac{\tau}{\Ts}$. The true delay $\tau_0$ is also scaled, defining $z_0  \triangleq \sfrac{\tau_0}{\Ts}$. The complex correlation as a function of $z$ and $\bm{x}$ is
\begin{align}
	\ACtilde &\triangleq \sum_{k \in \mathcal{K}} x^{\ast}[k] y[k]\exp\left(j2\pi z d[k] / K\right).
	\label{eq:complex_corr}
\end{align}
Let the total power of the \ac{ofdm} payload be $P = \sum_{k \in \mathcal{K}} x^{\ast}[k] x[k]$ and let the normalized received power at subcarrier $k$ be $\rho[k] = x^{\ast}[k] x[k]/P$. The total integrated \ac{snr} is then defined as $\gamma \triangleq gP/\sigma^2$. When normalized by the channel gain and total power, the complex correlation in (\ref{eq:complex_corr}) becomes
\begin{align}
	\frac{1}{\sqrt{g}P}\ACtilde = \ACf + \va,
\end{align}
where $\ACf$ is the normalized complex \ac{acf} of the received signal without noise, and $\va$ is the noise component. These are defined as
\begin{align}
	\ACf &\triangleq \sum_{k \in \mathcal{K}} \rho[k] \exp\left(j2\pi (z - z_0)d[k]/K + j\phi_0\right),\\
	\va &\triangleq \frac{1}{\sqrt{g}P}\sum_{k \in \mathcal{K}} x^{\ast}[k]v[k] \exp\left(j2\pi z d[k]/K\right).
\end{align}
It follows that $\va \sim \mathcal{CN}\left(0,\gamma^{-1}\right)$.

First consider the coherent reception case where the receiver knows $\phi_0$ exactly. In this case, the maximum-likelihood estimator becomes \cite{rife1974single}
\begin{align}
	\hat{z} &= \arg\max_{z}\; \Re\{\exp\left(-j\phi_0\right)\ACtilde\}\nonumber\\
	&= \arg\max_{z}\; \ACc + \Re\left\{\exp\left(-j\phi_0\right)\va\right\},
	\label{eq:mle_coh}
\end{align}
where $\ACc \triangleq \Re\left\{\exp\left(-j\phi_0\right)\ACf\right\}$ is the normalized coherent \ac{acf} of the received signal without noise, simplifying to \cite{laas2021ziv}
\begin{align}
	\ACc = \sum_{k \in \mathcal{K}} \rho[k] \cos(2\pi z d[k] / K).
	\label{eq:ac_coh}
\end{align}

Next consider the noncoherent reception case where $\phi_0$ is unknown. In this case, the maximum-likelihood estimator becomes \cite{rife1974single}
\begin{align}
	\hat{z} &= \arg\max_{z}\; \left|\ACtilde\right|^2\nonumber\\
	&= \arg\max_{z}\; \ACn + 2\Re\left\{\ACfAst\va\right\} + \left|\va\right|^2
	\label{eq:mle_incoh}
\end{align}
where $\ACn \triangleq |\ACf|^2$ is the normalized noncoherent \ac{acf} of the received signal without noise:
\begin{align}
	\ACn = &\left(\sum_{k \in \mathcal{K}}\rho[k] \cos{(2 \pi z d[k] / K)}\right)^2 \nonumber\\
	+ &\left(\sum_{k \in \mathcal{K}} \rho[k] \sin{(2 \pi z d[k] / K)}\right)^2.
	\label{eq:ac_incoh}
\end{align}

Both cases are of interest for radio positioning systems. The coherent detector minimizes the likelihood of detecting nearby sidelobes but may only be practical when the receiver's signal tracking has already obtained highly-accurate carrier phase measurements. Meanwhile, the noncoherent detector does not require \textit{a priori} knowledge of the carrier phase, making it more applicable when accurate carrier phase measurements are not available. Much of the prior work on analyzing and optimizing the \ac{zzb} for \ac{toa} estimation has neglected to distinguish these two cases.

\subsection{Estimation Error Bounds}

Letting $\hat{\tau}$ be an unbiased estimate of $\tau_0$, the \ac{crlb} on \ac{toa} estimation error variance for \ac{ofdm} signals is given by \cite{xu2016maximum}
\begin{align}
	\mathbb{E}\left[(\hat{\tau} - \tau_0)^2\right] &\geq \sigma_{\text{CRLB}}^2(\bm{\rho})\nonumber\\
	&\triangleq \left( 8\pi^2 \gamma \Delta_{\text{f}}^2 \sum_{k \in \mathcal{K}} d^{2}[k] \rho[k] \right)^{-1}.
\end{align}
%
The \ac{crlb} provides a valuable insight, namely, that estimator error variance scales inversely with the square of the distance (in frequency) of powered subcarriers from the carrier at the center of the band. Thus, a \ac{crlb}-optimal strategy would allocate all power to the extremes: $\rho[k] = 1/2$ for $k \in {K/2 - 1, K/2}$, otherwise $\rho[k] = 0$.  However, the \ac{crlb} is only applicable at high \ac{snr} because it ignores the thresholding effects caused by sidelobes at lower \ac{snr} \cite{graff2024purposeful}.


To resolve the shortcomings of the \ac{crlb}, the \ac{zzb} is computed to characterize the \ac{toa} error variance for a given \ac{ofdm} pilot power allocation. This is preferred over the \ac{crlb} because it accounts for thresholding effects caused by sidelobes in the signal's \ac{acf}. The \ac{zzb} considers the error probability of a binary detection problem with equally-likely hypotheses: (1) the received signal experienced delay $\tau_0$, and (2) the received signal experienced delay $\tau_0 + \tau$. Assuming this error probability is shift-invariant, the hypotheses can be simplified without loss of generality by assuming $\tau_0 = 0$. The minimum probability of error for this problem is denoted $\Pmin$.

Assuming \textit{a priori} knowledge that the \ac{toa} is uniformly distributed in $[0,\Ta]$, the \ac{zzb} can be defined as \cite{Ziv1969,dardari2009ziv}
\begin{align}
	\mathbb{E}\left[(\hat{\tau} - \tau_0)^2\right]  &\geq \ZZBRho \nonumber\\
	&\triangleq \frac{1}{\Ta}\int_{0}^{\Ta} \tau (\Ta - \tau) P_{\text{min}}(\tau/T_s,\bm{\rho}) d\tau \nonumber\\
	&= \frac{\Ts^2}{\Na}\int_{0}^{\Na} z (\Na - z) \Pmin dz,
	\label{eq:zzb}
\end{align}
where $\Na \triangleq \sfrac{\Ta}{\Ts}$. Expressing the \ac{zzb} in this form with normalized time delays is useful since no terms in the integrand depend on the subcarrier spacing $\Delta_{\text{f}}$, allowing the bound to be easily scaled to account for varying $\Delta_{\text{f}}$.

As in the previous section, the receiver may perform either coherent or noncoherent detection, each having different error probabilities. The minimum probability of error for the coherent detector is \cite{proakis2008digital,dardari2009ziv}
\begin{align}
	\PminC &= Q\left(\sqrt{\gamma\left(1 - \ACc\right)}\right)\nonumber\\
	&= \frac{1}{2} - \frac{1}{2} \text{erf}\left(\sqrt{ \frac{\gamma}{2} \left( 1 - \ACc \right) } \right),
	\label{eq:Pmin_coh}
\end{align}
and the resulting coherent \ac{zzb} is defined as
\begin{align}
	\ZZBcRho \triangleq \frac{\Ts^2}{\Na}\int_{0}^{\Na} z (\Na - z) \PminC dz.
	\label{eq:zzb_coh}
\end{align}
Meanwhile, the minimum probability of error for the noncoherent detector is \cite{proakis2008digital,helstrom1955resolution}
\begin{align}
	\PminN &= Q_{1}(a,b) - \frac{1}{2}\exp{\left(\sfrac{-(a^2 + b^2)}{2}\right)}I_0(ab),\label{eq:Pmin_incoh}\\
	a &\triangleq \sqrt{\frac{\gamma}{2}\left(1-\sqrt{1-\ACn}\right)},\label{eq:a}\\
	b &\triangleq \sqrt{\frac{\gamma}{2}\left(1+\sqrt{1-\ACn}\right)}, \label{eq:b}
\end{align}
and the resulting noncoherent \ac{zzb} is defined as
\begin{align}
	\ZZBnRho \triangleq \frac{\Ts^2}{\Na}\int_{0}^{\Na} z (\Na - z) \PminN dz.
	\label{eq:zzb_incoh}
\end{align}

It is important to note that neither the coherent nor noncoherent expressions for error probability depends on the phase of $x[k]$, only on the power $\rho[k]$ allocated to each subcarrier. The next section explores the optimization of the power allocations $\rho[k]$ to minimize the \ac{zzb} in both schemes.

\section{Optimization}
\label{section:optim}


\subsection{Convex-Constrained Problem}
\label{sec:convex_constrained}
\ac{zzb}-optimal \ac{ofdm} power allocations can be found by minimizing the \ac{zzb} subject to a fixed power budget constraint. This optimization problem is written as
\begin{align}
	\min_{\bm{\rho}} \quad & \ZZBRho \nonumber\\
	\textrm{s.t.} \quad & \rho[k] \geq 0, \quad k \in \mathcal{K}\nonumber\\
	& \bm{1}^T\bm{\rho} - 1 = 0
	\label{eq:zzb_optim}
\end{align}
where $\ZZBRho$ is substituted with $\ZZBcRho$ for the coherent case or $\ZZBnRho$ for the noncoherent case. The affine equality constraint can be removed by expressing $\bm{\rho}$ as an affine function of a $K-1$ dimensional vector $\tilde{\bm{\rho}}$. Let $\tilde{\rho}[k] = \rho[k]$ for $k \in \tilde{\mathcal{K}}$, where $\tilde{\mathcal{K}} \triangleq \mathcal{K}\setminus\{0\}$ is the set of subcarrier indices excluding the DC subcarrier. Then define $\bm{F} \in \mathbb{R}^{K\times{}K-1}$ as $\bm{F} \triangleq \left[-\bm{1},\;\bm{I}\right]^{T}$ and define $\bm{e}_{0} \in \mathbb{R}^{K}$ as a one-hot vector with a $1$ at index $0$ and $0$ elsewhere. Then, under the affine equality constraint, $\bm{\rho}$ can be parameterized as $\bm{\rho} = f(\tilde{\bm{\rho}}) \triangleq \bm{F}\tilde{\bm{\rho}} + \bm{e}_{0}$. The optimization problem in (\ref{eq:zzb_optim}) can then be reframed as
\begin{align}
	\min_{\tilde{\bm{\rho}}} \quad & \ZZBRhoR \nonumber\\
	\textrm{s.t.} \quad & \tilde{\rho}[k] \geq 0,\quad k \in \tilde{\mathcal{K}}\nonumber\\
	&\bm{1}^T\tilde{\bm{\rho}} - 1\leq 0
	\label{eq:zzb_optim_reframe}
\end{align}
where the last inequality constraint arises from the inequality constraint on $\rho[0]$ in (\ref{eq:zzb_optim}). Parameterizing the problem in this manner ensures that the gradients can be computed. Otherwise, the integrals for computing the gradients may diverge. The following section will prove that the optimization problem (\ref{eq:zzb_optim}) is convex for both the coherent and noncoherent error probabilities. The following proof of convexity for the coherent case in Theorem~\ref{thm:convexity_coh} is different from the proof in \cite{xiong2023snr} and demonstrates convexity with respect to \ac{ofdm} subcarrier power allocations.


\begin{theorem}
	The coherent \ac{zzb} $\ZZBcRho$ is convex with respect to the \ac{ofdm} power allocations $\bm{\rho}$ on the domain $\left\{\bm{\rho} \;|\; \bm{1}^{T}\bm{\rho} \leq 1, \rho[k] \geq 0\right\}$.
	\label{thm:convexity_coh}
\end{theorem}

\begin{proof}
For ease of expression, begin by defining
\begin{align}
	\lambda &\triangleq \ACc,
\end{align}
resulting in
\begin{align}
	\PminC = Q\left(\sqrt{\gamma\left( 1 - \lambda \right)}\right).
\end{align}
Note that $0 \leq \lambda \leq 1$. By Craig's formula \cite{craig1991new}, which is valid for $\sqrt{\gamma(1-\lambda)} \geq 0$ or equivalently $\lambda \leq 1$, $\PminC$ is rewritten as
\begin{align}
	&\PminC = \frac{1}{\pi} \int_{0}^{\sfrac{\pi}{2}} \exp\left(\frac{-\gamma(1-\lambda)}{2\sin^2(\theta)}\right) d\theta \nonumber\\
	&\quad= \lim_{\epsilon\to0} \frac{1}{\pi} \int_{\epsilon}^{\sfrac{\pi}{2}} \exp\left(\frac{-\gamma}{2\sin^2(\theta)}\right) \exp\left(\frac{\gamma\lambda}{2\sin^2(\theta)}\right) d\theta \nonumber\\
	&\quad= \lim_{\epsilon\to0} \frac{1}{\pi} \int_{\epsilon}^{\sfrac{\pi}{2}} C_{1}(\theta) \exp\left(C_{2}(\theta)\lambda\right) d\theta.
\end{align}
For any $\theta$, it follows that $C_{1}(\theta) \geq 0$, $C_{2}(\theta) \geq 0$, and $\exp\left(C_{2}(\theta)\lambda\right)$ is convex with respect to $\lambda$. Since non-negative weighted integrals preserve convexity \cite{boyd2004convex}, $\PminC$ is convex with respect to $\lambda$. Recognizing that $\ACc$ is an affine mapping of $\bm{\rho}$ and that composition with an affine map preserves convexity \cite{boyd2004convex}, it follows that $\PminC$ is convex with respect to $\bm{\rho}$. Finally, the integral in (\ref{eq:zzb_coh}) preserves convexity. Therefore $\ZZBcRho$, the \ac{zzb} defined in (\ref{eq:zzb_coh}), is convex with respect to $\bm{\rho}$ on this domain.
\end{proof}


This convexity is preserved under the affine mapping used to remove the equality constraint in (\ref{eq:zzb_optim}), ensuring that (\ref{eq:zzb_optim_reframe}) is also a convex problem. This convexity allows (\ref{eq:zzb_optim_reframe}) to be efficiently solved using existing solvers. To aid the implementation of first and second order methods with these solvers, the gradient and Hessian of the \ac{zzb} with respect to the power allocations $\tilde{\bm{\rho}}$ are provided. When expressing the partial derivatives, the shorthand notation $\tilde{\rho}_{n}$ is substituted in place of $\tilde{\rho}[n]$.

The gradient of $\PminCR$ with respect to $\tilde{\bm{\rho}}$ is constructed of entries
\begin{align}
		\frac{\partial}{\partial \tilde{\rho}_n} \PminCR &= \frac{\sqrt{\gamma}}{2\sqrt{2\pi}} \exp \left( \frac{-\gamma}{2} \left( 1 - \ACcR \right) \right) \nonumber\\
		&\times \frac{\cos(2\pi z d[n{+}1]/K) - 1}{\sqrt{1 - \ACcR}}.
		\label{eq:grad_pmin_coh}
\end{align}
Similarly, the Hessian matrix of $\PminCR$ with respect to $\tilde{\bm{\rho}}$ is constructed of entries
\begin{align}
	&\frac{\partial^2}{\partial \tilde{\rho}_n \partial \tilde{\rho}_m}  \PminCR \nonumber\\
	&= \frac{\sqrt{\gamma}}{4\sqrt{2\pi}} \exp \left( \frac{-\gamma}{2} \left( 1 - \ACcR \right) \right) \nonumber\\
	&\times \frac{(\cos{(2\pi z d[n{+}1] / K)}-1)(\cos{(2\pi z d[m{+}1] / K)}-1)}{\sqrt{1 - \ACc}} \nonumber\\
	&\times \left(\gamma + \left(1 - \ACc\right)^{-1} \right).
	\label{eq:hess_pmin_coh}
\end{align}
The expressions for the gradient in (\ref{eq:grad_pmin_coh}) and Hessian in (\ref{eq:hess_pmin_coh}) are derived in Appendix~\ref{sec:app_coh_grad_hess}. Using (\ref{eq:grad_pmin_coh}), the gradient of the ZZB with respect to $\tilde{\bm{\rho}}$ can then be constructed with entries
\begin{align}
	&\frac{\partial}{\partial \tilde{\rho}_n} \ZZBcRhoR \nonumber\\
	&= \frac{\Ts^2}{\Na}\int_{0}^{\Na} z (\Na - z) \left(\frac{\partial}{\partial \tilde{\rho}_n} \PminCR\right) dz.
	\label{eq:grad_zzb_coh}
\end{align}
Similarly using (\ref{eq:hess_pmin_coh}), the Hessian matrix of the ZZB with respect to $\tilde{\bm{\rho}}$ can be constructed with entries
\begin{align}
	&\frac{\partial^2}{\partial \tilde{\rho}_n \partial \tilde{\rho}_m} \ZZBcRhoR \nonumber\\
	&= \frac{\Ts^2}{\Na}\int_{0}^{\Na} z (\Na - z) \left(\frac{\partial^2}{\partial \tilde{\rho}_n \partial \tilde{\rho}_m} \PminCR\right) dz.
	\label{eq:hess_zzb_coh}
\end{align}


\begin{theorem}
	The noncoherent \ac{zzb} $\ZZBnRho$ is convex with respect to the \ac{ofdm} power allocations $\bm{\rho}$ on the domain $\left\{\bm{\rho} \;|\; \bm{1}^{T}\bm{\rho} \leq 1, \rho[k] \geq 0\right\}$.
	\label{thm:convexity_incoh}
\end{theorem}

\begin{proof}
For ease of expression, begin by defining
\begin{align}
	\lambda &\triangleq \ACn,
\end{align}
and note the following equivalences for $a$ and $b$ defined in (\ref{eq:a}) and (\ref{eq:b})\cite{helstrom1955resolution}:
\begin{align}
	\frac{a}{b} = \frac{\sqrt{\lambda}}{1+\sqrt{1-\lambda}}, \quad
	ab = \frac{\gamma}{2}\sqrt{\lambda}, \quad a^2 + b^2 = \frac{\gamma}{2}.
\end{align}

Then express the Marcum-Q function using its Neumann series expansion \cite{annamalai2008simple} as
\begin{align}
	Q_{1}(a,b) = \exp{\left(\sfrac{-(a^2 + b^2)}{2}\right)} \sum_{n=0}^{\infty} \left(\frac{a}{b}\right)^{n} I_{-n}(ab).
	\label{eq:marcum_Q}
\end{align}
Recalling that $I_{-n}(\cdot) = I_{n}(\cdot)$ for integer values of $n$, this expansion can be substituted into $\PminN$, resulting in
\begin{align}
	\PminN &= \frac{1}{2}\exp{\left(\sfrac{-\gamma}{2}\right)}\overbrace{I_{0}\left(\sfrac{\gamma}{2}\sqrt{\lambda}\right)}^{A} \nonumber\\
	 &+\exp{\left(\sfrac{-\gamma}{2}\right)} \sum_{n=1}^{\infty} \underbrace{\left(\frac{\sqrt{\lambda}}{1+\sqrt{1-\lambda}}\right)^{n} I_{n}\left(\sfrac{\gamma}{2}\sqrt{\lambda}\right)}_{B}.
\end{align}
Note that the $n=0$ term of the sum in (\ref{eq:marcum_Q}) was combined with the original modified Bessel function term in (\ref{eq:Pmin_incoh}). Since the exponential terms are positive, $\PminN$ is a positive weighted sum of $A$ and $B$ and is convex if both $A$ and $B$ are convex. Now consider the series definition of the modified Bessel function of the first kind:
\begin{align}
	I_{n}(z) = &\left(\frac{1}{2}z\right)^{n} \sum_{l=0}^{\infty} \frac{\left(\frac{1}{4}z^2\right)^l}{l!\Gamma(n+l+1)} \nonumber\\
	= &\sum_{l=0}^{\infty} C_{3}(n,l) z^{2l+n}.
\end{align}
Then $A$ can be expressed as
\begin{align}
	A = I_{0}\left(\sfrac{\gamma}{2}\sqrt{\lambda}\right) = \sum_{l=0}^{\infty} C_{3}(0,l) \lambda^{l}.
\end{align}
Since $C_{3}(0,l)$ is positive and $\lambda^{l}$ is a convex function of $\lambda$ for $0 \leq \lambda \leq 1$, it follows that $A$ is convex with respect to $\lambda$ for $0 \leq \lambda \leq 1$. Now $B$ can be expressed as
\begin{align}
	B = &\sum_{l=0}^{\infty} C_{3}(n,l) \left(\frac{\sqrt{\lambda}}{1+\sqrt{1-\lambda}}\right)^{n} \left(\frac{\gamma}{2}\right)^{2l+n} \lambda^{l+n/2} \nonumber\\
	= &\sum_{l=0}^{\infty} C_{4}(n,l) \left(1+\sqrt{1-\lambda}\right)^{-n} \lambda^{l+n}.
\end{align}

Both $\left(1+ \sqrt{1-\lambda}\right)^{-n}$ and $\lambda^{l+n}$ are convex, nondecreasing, nonegative, and bounded for all $0 \leq \lambda \leq 1$, $n \geq 1$, and $l \geq 0$. Since $C_{4}(n,l)$ is positive, it follows that $B$ is convex for $0 \leq \lambda \leq 1$. Since both $A$ and $B$ are convex for $0 \leq \lambda \leq 1$, $\PminN$ is convex with respect to $\lambda$ on the domain $0 \leq \lambda \leq 1$. It also follows that $\PminN$ is nondecreasing with respect to $\lambda$ since both $A$ and $B$ are nondecreasing on this domain.

$\PminN$ will now be shown to be convex with respect to $\rho[k]$ when $\ACn$ is substituted for $\lambda$ and the power allocations are constrained to the convex domain $\left\{\bm{\rho} \;| \;\bm{1}^{T}\bm{\rho} \leq 1, \rho[k] \geq 0\right\}$. On this domain, it follows that $0 \leq \ACn \leq 1$, which is a convex set equivalent to the domain on which $\PminN$ was proven to be convex and nondecreasing with respect to $\ACn$. Additionally, $\ACn$ is a convex function of $\bm{\rho}$. Therefore, by composition, $\PminN$ is convex with respect to $\bm{\rho}$ on the domain $\left\{\bm{\rho} \;| \;\bm{1}^{T}\bm{\rho} \leq 1, \rho[k] \geq 0\right\}$. Finally, the integral in (\ref{eq:zzb_incoh}) preserves convexity. Therefore $\ZZBnRho$, the \ac{zzb} defined in (\ref{eq:zzb_incoh}), is convex with respect to $\bm{\rho}$ on this domain.
\end{proof}

As in the coherent case, this convexity is preserved under the affine mapping used to remove the equality constraint in (\ref{eq:zzb_optim}), ensuring that (\ref{eq:zzb_optim_reframe}) is a convex problem. Gradients and Hessians may also be computed when solving the noncoherent form of (\ref{eq:zzb_optim}). Due to its complexity, the Hessian for the noncoherent case is omitted in this paper, and first-order or quasi-Newton methods are recommended for solving the noncoherent optimization. The gradient of $\PminNR$ with respect to the power allocations $\tilde{\bm{\rho}}$ is constructed of entries
\begin{align}
	&\frac{\partial}{\partial \tilde{\rho}_n} \PminNR \nonumber\\
	&= \left(\frac{\partial}{\partial \tilde{\rho}_n} Q_{1}(a,b)\right)
	+ \left(\frac{\partial}{\partial \tilde{\rho}_n} \frac{-1}{2}\exp{\left(\sfrac{-\gamma}{2}\right)}I_{0}(ab)\right),
	\label{eq:grad_pmin_incoh}
\end{align}
whose terms are derived in detail in Appendix~\ref{sec:app_noncoh_grad}. Using (\ref{eq:grad_pmin_incoh}), the gradient of the \ac{zzb} with respect to $\tilde{\bm{\rho}}$ can then be constructed with entries
\begin{align}
	&\frac{\partial}{\partial \tilde{\rho}_n} \ZZBnRhoR \nonumber\\
	&= \frac{\Ts^2}{\Na}\int_{0}^{\Na} z (\Na - z) \left(\frac{\partial}{\partial \tilde{\rho}_n} \PminNR \right) dz.
	\label{eq:grad_zzb_incoh}
\end{align}

\subsection{Integer-Constrained Problem}
\label{sec:integer_constrained}

The results obtained from the optimization in (\ref{eq:zzb_optim}) may not always be practical since power is allocated across the entire set of available subcarriers. In a dual-functional system, resources must also be allocated for communications purposes, and it may be preferable to multiplex these resources in frequency. Therefore, it is desirable to further constrain the optimization problem to the selection of a set of $L$ subcarriers. The optimization problem is also further restricted to equal allocation of power across these $L$ subcarriers, simplifying the optimization and avoiding narrowband interference concerns caused by allocating significant power into a small number of subcarriers. Under these constraints, the optimization problem can be written as
\begin{align}
	\min_{\bm{\rho}} \quad & \ZZBRho \nonumber\\
	\textrm{s.t.} \quad & \rho[k] \in \left\{0,\frac{1}{L}\right\}, \quad k \in \mathcal{K}\nonumber\\
	& \bm{1}^T\bm{\rho} = 1
	\label{eq:zzb_int_optim}
\end{align}

This is an NP-hard nonlinear integer programming problem. For any reasonable number of subcarriers $K$, a brute-force search is computationally prohibitive and requires evaluating ${K \choose L}$ selections. In light of these difficulties, a branch-and-bound algorithm \cite{tuy1998convex} is proposed.

Recognizing that (\ref{eq:zzb_optim}) is a convex relaxation of (\ref{eq:zzb_int_optim}), the solutions to the convex-constrained problem in (\ref{eq:zzb_optim}) serve as a lower bound to the integer-constrained problem in (\ref{eq:zzb_int_optim}). Furthermore, the problem can be divided into subproblems in which certain binary decisions have already been made through the addition of equality constraints. In particular, consider two disjoint sets of subcarrier indices: $\Kz \subseteq \mathcal{K}$ for which the power allocations are set to $0$ and $\Ko \subseteq \mathcal{K}$ for which the power allocations are set to $1/L$. For all subcarriers not included in either of these two sets, $k \in \mathcal{K} \setminus \Kz \bigcup \Ko$, power allocations can be optimized without integer constraints to form a lower bound on each subproblem. Adding these equality constraints results in the relaxed subproblem
\begin{align}
	\min_{\bm{\rho}} \quad & \ZZBRho \nonumber\\
	\textrm{s.t.} \quad & \rho[k] \geq 0, \quad k \in \mathcal{K}\nonumber\\
	& \bm{1}^T\bm{\rho} = 1\nonumber\\
	& \rho[k] = 0,\quad k \in \Kz\nonumber\\
	& \rho[k] = \frac{1}{L},\quad k \in \Ko
	\label{eq:zzb_optim_relax}
\end{align}

When $\Kz \bigcup \Ko = \emptyset$, this simplifies to (\ref{eq:zzb_optim}), and when $|\Ko| > L$, the problem is infeasible. This relaxed subproblem allows the branch and bound algorithm to determine lower bounds on (\ref{eq:zzb_int_optim}). Upper bounds on (\ref{eq:zzb_int_optim}) are then found by evaluating the \ac{zzb} on a feasible approximation of the solutions of the relaxed subproblem in (\ref{eq:zzb_optim_relax}). These feasible approximations are found by setting the $L$ subcarriers with the highest power allocations to a power of $1/L$ and setting all other subcarriers to a power of zero. This rounding function is denoted $\max_L(\bm{\rho})$.

With these upper and lower bounds, Algorithm \ref{alg:branch_and_bound} finds approximate solutions for (\ref{eq:zzb_int_optim}). It adds subproblems to a priority queue, sorted such that subproblems with the smallest lower bound are popped first. Each subproblem has a relaxed solution $\check{\bm{\rho}}$, a lower bound $\check{\text{LB}}$ achieved by the relaxed solution, and constraint sets $\KzC$ and $\KoC$. In each iteration, a subproblem is popped from the queue and a subcarrier index $\kbranch$ is selected for branching. Branching creates two additional subproblems: one where the power at subcarrier $\kbranch$ is set to zero and one where the power at subcarrier $\kbranch$ is set to $1/L$. For each subproblem, a relaxed solution $\bm{\rho}$, a lower bound $\text{LB}$, and an upper bound are computed. Each subproblem is then inserted into the priority queue so long as its lower bound is less than the tightest upper bound $\text{UB}$, after which the algorithm iterates. Since there are no guarantees that this algorithm converges to the optimal solution in fewer iterations than the brute force search, the iterations continue until either the optimality gap $\delta \triangleq \left(\text{UB} - \check{\text{LB}}\right)/\check{\text{LB}}$ is less than the tolerable gap $\delta_{\text{tol}}$ or the number of iterations exceeds $N_{\text{iter}}$.

\begin{algorithm}
	\caption{Branch and Bound}
	\label{alg:branch_and_bound}
	\small
	\hspace*{\algorithmicindent}\textbf{Input}: $N_{\text{iter}},\delta_{\text{tol}}$\\
	\hspace*{\algorithmicindent}\textbf{Output}: $\bm{\rho}^{\ast}$
	
	\begin{algorithmic}[1]
		\STATE $\Kz \leftarrow \emptyset$, $\Ko \leftarrow \emptyset$ \COMMENT{\textbf{Initialize}}
		\STATE $\bm{\rho} \leftarrow \text{Solution to (\ref{eq:zzb_optim})}$
		\STATE $\text{LB} \leftarrow \ZZBRho,\quad\text{UB} \leftarrow \ZZBRhoMaxL$
		\STATE $\bm{\rho}^{\ast} \leftarrow \max_L(\bm{\rho})$
		\STATE $\text{iter} \leftarrow 0,\quad\delta \leftarrow \infty$
		\STATE Insert $\{\bm{\rho},\text{LB},\Kz,\Ko\}$ into \texttt{PriorityQueue}
		\WHILE{$\text{iter} < N_{\text{iter}}$ \OR $\delta < \delta_{\text{tol}}$}
			\STATE Pop $\{\check{\bm{\rho}},\check{\text{LB}},\KzC,\KoC\}$ from \texttt{PriorityQueue}
			\STATE $\delta = \frac{\text{UB} - \check{\text{LB}}}{\check{\text{LB}}}$
			\STATE $\kbranch = \arg\min_{k \in \mathcal{K}\setminus\{\KzC \bigcup \KoC\}} |\check{\rho}[k] - \frac{1}{2L}|$ \COMMENT{\textbf{Branch}}
			\STATE $\Kz \leftarrow \KzC \bigcup \{\kbranch\}$, $\Ko \leftarrow \KoC$ \COMMENT{\textbf{Subproblem 1}}
			\STATE $\bm{\rho} \leftarrow \text{Solution to (\ref{eq:zzb_optim_relax})}$
			\STATE $\text{LB} \leftarrow \ZZBRho$
			\IF{$\ZZBRhoMaxL < \text{UB}$}
				\STATE $\text{UB} \leftarrow \ZZBRhoMaxL$
				\STATE $\bm{\rho}^{\ast} \leftarrow \max_L(\bm{\rho})$
			\ENDIF
			\IF{$\text{LB} < \text{UB}$}
				\STATE Insert $\{\bm{\rho},\text{LB},\Kz,\Ko\}$ into \texttt{PriorityQueue}
			\ENDIF
			\STATE $\Kz \leftarrow \KzC$, $\Ko \leftarrow \KoC \bigcup \{\kbranch\}$ \COMMENT{\textbf{Subproblem 2}}
			\STATE $\bm{\rho} \leftarrow \text{Solution to (\ref{eq:zzb_optim_relax})}$
			\STATE $\text{LB} \leftarrow \ZZBRho$
			\IF{$\ZZBRhoMaxL < \text{UB}$}
				\STATE $\text{UB} \leftarrow \ZZBRhoMaxL$
				\STATE $\bm{\rho}^{\ast} \leftarrow \max_L(\bm{\rho})$
			\ENDIF
			\IF{$\text{LB} < \text{UB}$}
				\STATE Insert $\{\bm{\rho},\text{LB},\Kz,\Ko\}$ into \texttt{PriorityQueue}
			\ENDIF
			\STATE $\text{iter} \leftarrow$ \text{iter} + 1
		\ENDWHILE
	\end{algorithmic}
\end{algorithm}


\section{Numerical Results}
\label{sec:num_results}

\subsection{Method and Setup}

In a numerical demonstration, \ac{zzb}-optimal power allocations were found using both the convex-constrained  problem (\ref{eq:zzb_optim}) and the integer-constrained problem (\ref{eq:zzb_int_optim}), and each problem was solved for both coherent reception, using the error probability in (\ref{eq:Pmin_coh}), and noncoherent reception, using the error probability in (\ref{eq:Pmin_incoh}). These four variations of optimized allocations will be referred to as the coherent convex-optimized, coherent integer-optimized, noncoherent convex-optimized, and noncoherent integer-optimized allocations.
For all problems, the \ac{ofdm} signal was configured with $K = 64$ subcarriers and a subcarrier spacing of $\Delta_{\text{f}} = \SI{15.625}{\kilo\hertz}$, yielding a symbol duration of \SI{64}{\micro\second}. The receiver's \textit{a priori} \ac{toa} duration was $\Na = 16$ samples or $\Ta  = \SI{16}{\micro\second}$.

The coherent convex-constrained problem (\ref{eq:zzb_optim}) was solved using an interior-point method using the analytical Jacobian and Hessian as provided by (\ref{eq:grad_zzb_coh}) and (\ref{eq:hess_zzb_coh}), iterating until convergence or a maximum of $30$ iterations. The noncoherent version was solved similarly using the gradient in (\ref{eq:grad_zzb_incoh}) and an approximate Hessian computed using BFGS \cite{fletcher2000practical}.

The integer-constrained problem (\ref{eq:zzb_int_optim}) was solved using the branch-and-bound algorithm as outlined in Algorithm~\ref{alg:branch_and_bound}, solving each relaxed subproblem (\ref{eq:zzb_optim_relax}) using an identical strategy as with the convex-constrained problem. The power allocations were restricted to $L = 8$ subcarriers. For both the coherent and noncoherent schemes, the branch-and-bound algorithm was configured with a tolerance $\delta_{\text{tol}} = 0.01$ and a maximum number of iterations $N_{\text{iter}} = 2000$.


\begin{figure*}[h!]
	\centering
	\begin{subfigure}{0.245\linewidth}
		\centering
		\includegraphics[width=\textwidth]{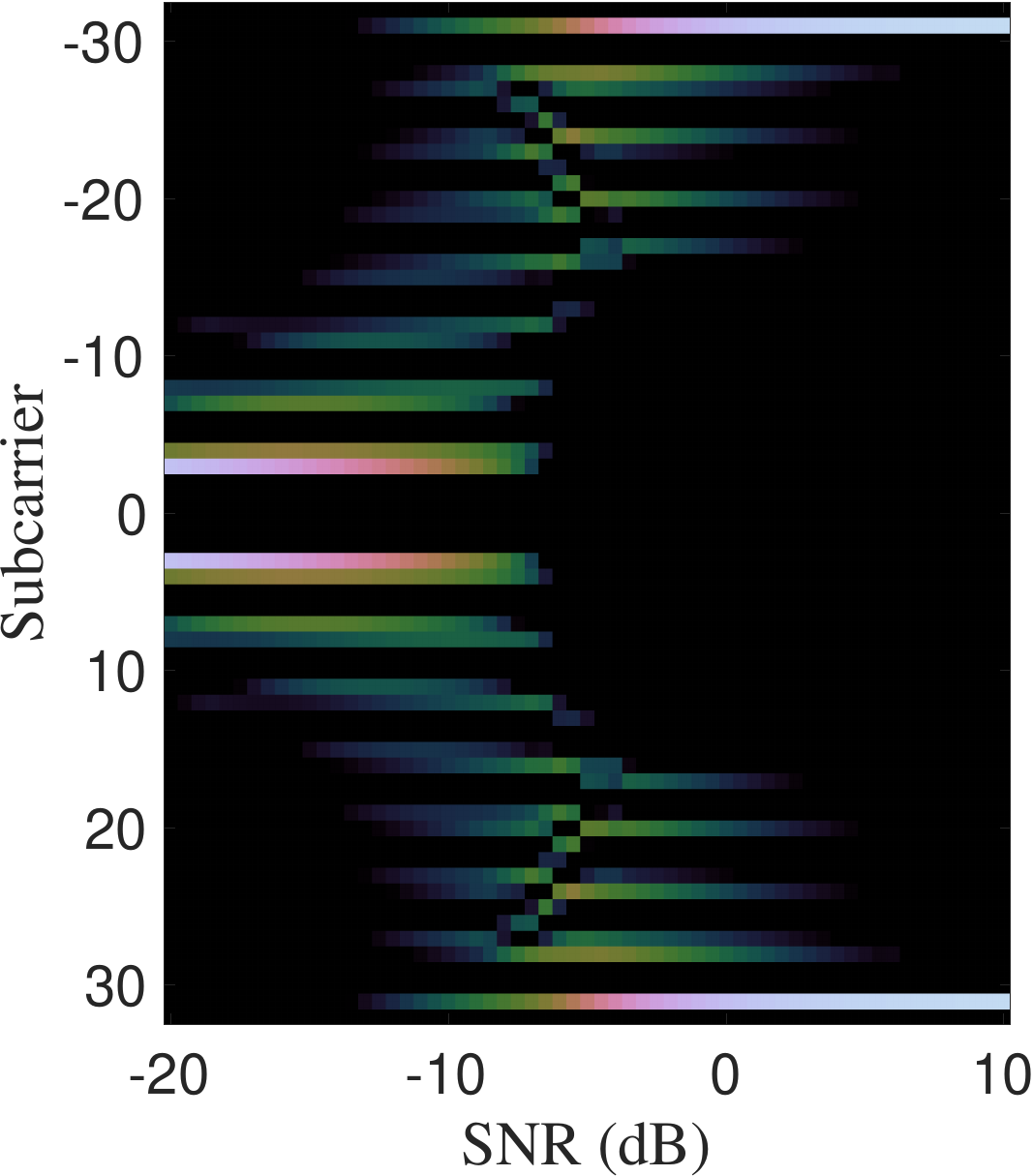}
		\caption{}
		\label{fig:rho_coh}
	\end{subfigure}
	\begin{subfigure}{0.235\linewidth}
		\centering
		\includegraphics[width=\textwidth]{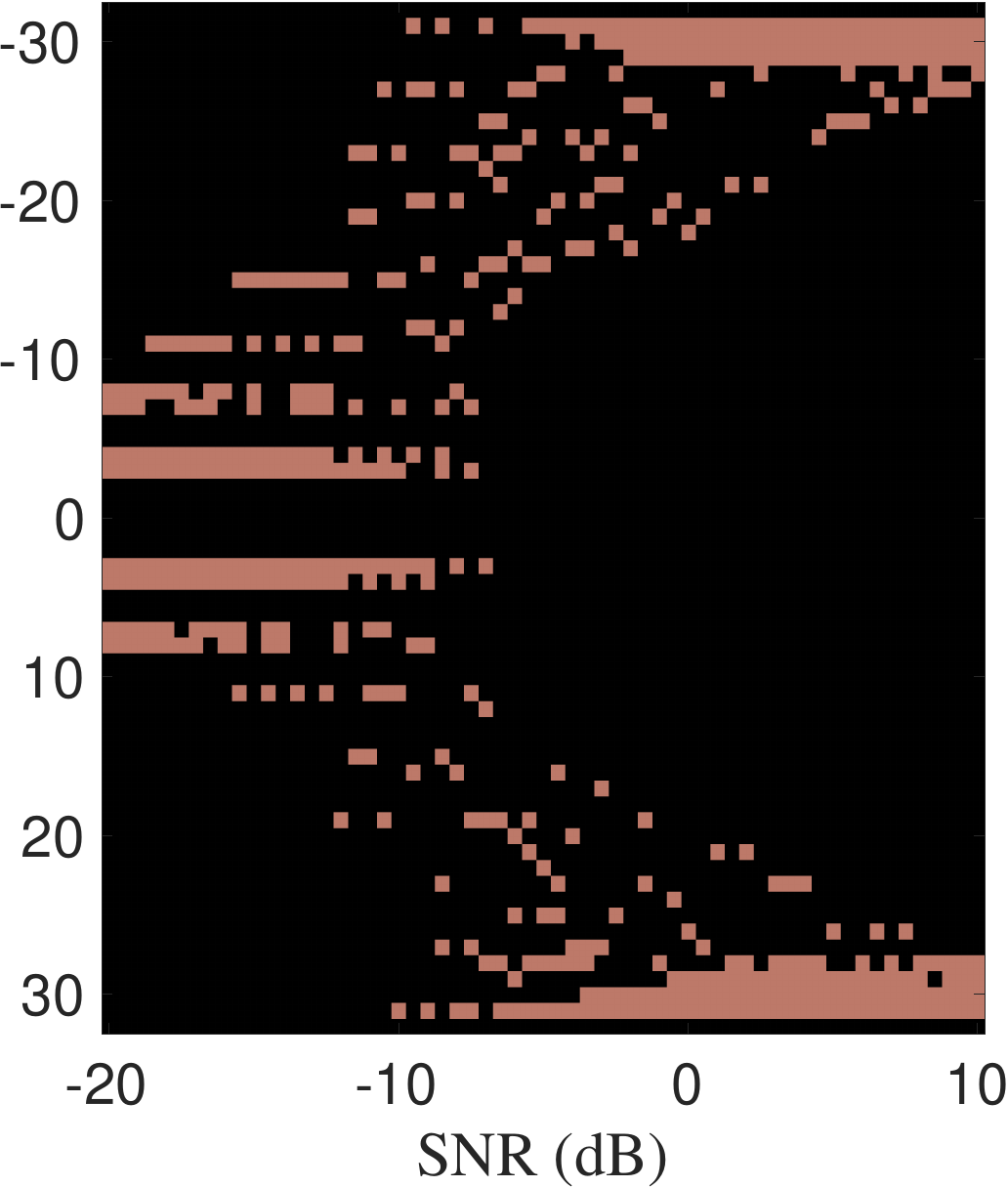}
		\caption{}
		\label{fig:rho_int_coh}
	\end{subfigure}
	\begin{subfigure}{0.235\linewidth}
		\centering
		\includegraphics[width=\textwidth]{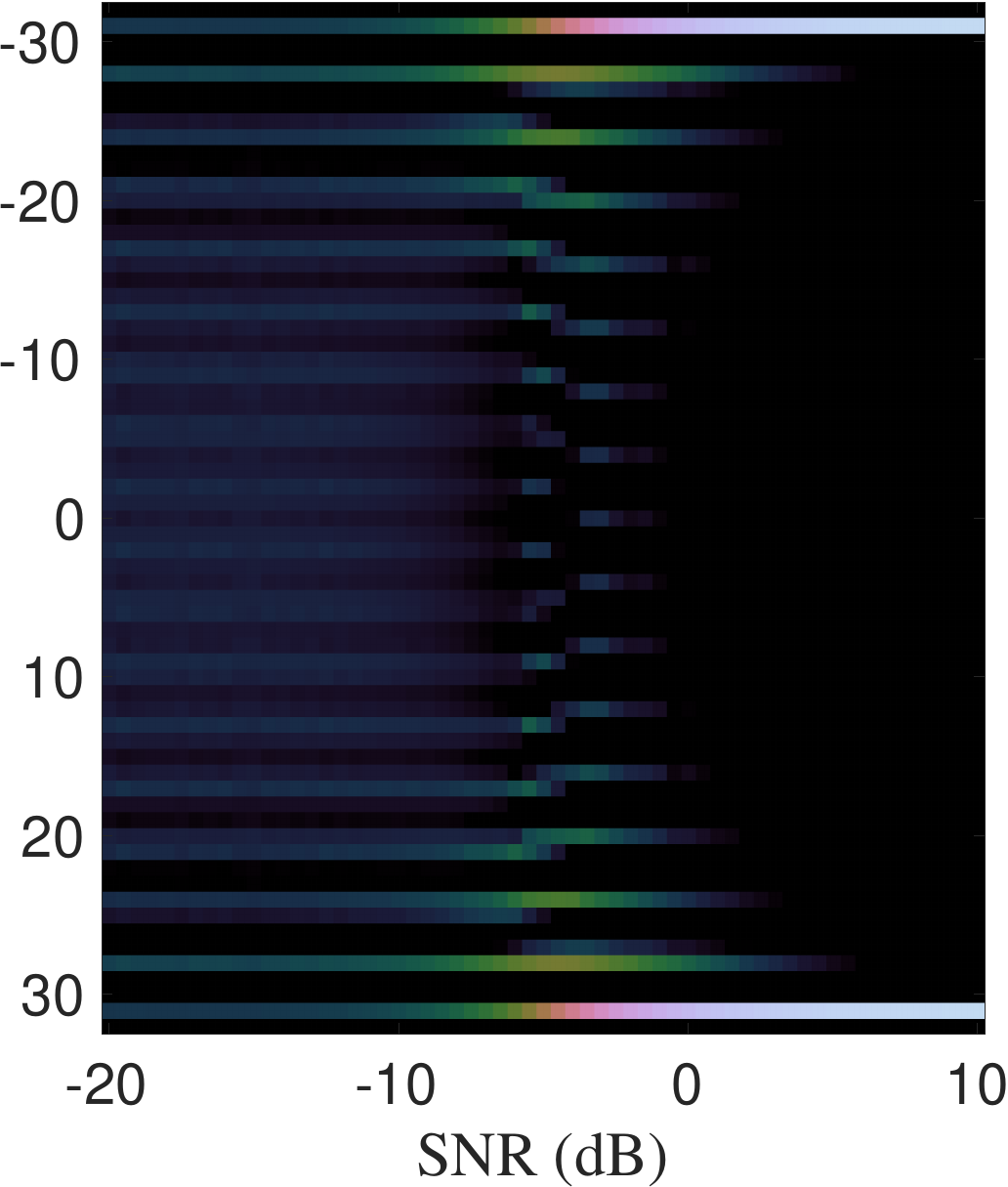}
		\caption{}
		\label{fig:rho_incoh}
	\end{subfigure}
	\begin{subfigure}{0.26\linewidth}
		\centering
		\includegraphics[width=\textwidth]{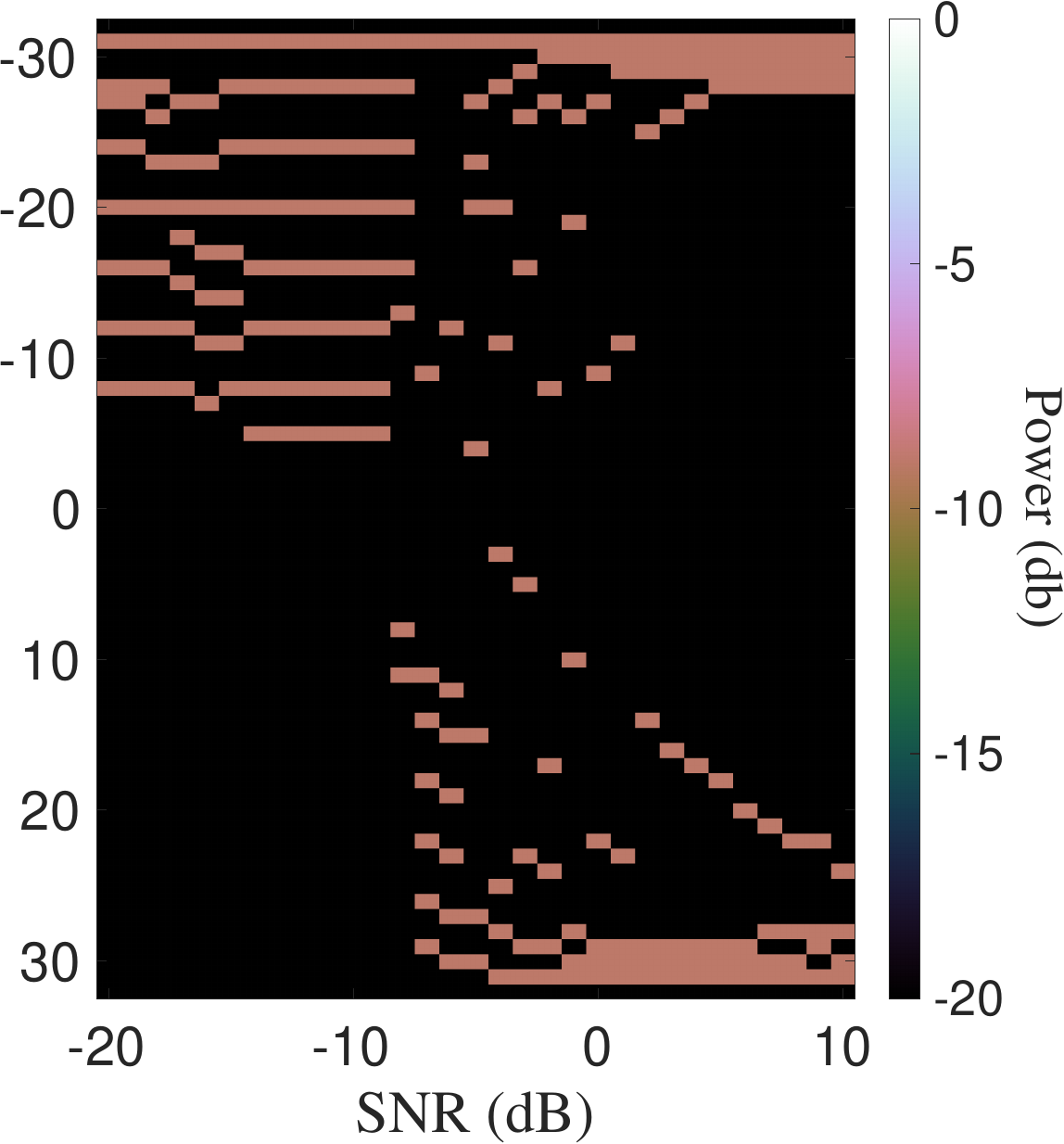}
		\caption{}
		\label{fig:rho_int_incoh}
	\end{subfigure}
	\caption{The optimized subcarrier power allocations over a range of \acp{snr} for the coherent convex-constrained problem (a), the coherent integer-constrained problem (b), the noncoherent convex-constrained problem (c), and the noncoherent integer-constrained problem (d). For the integer-constrained problems, $L=8$.}
	\label{fig:rho_list}
\end{figure*}

\begin{figure*}[h!]
	\centering
	\begin{subfigure}{0.245\linewidth}
		\centering
		\includegraphics[width=\textwidth]{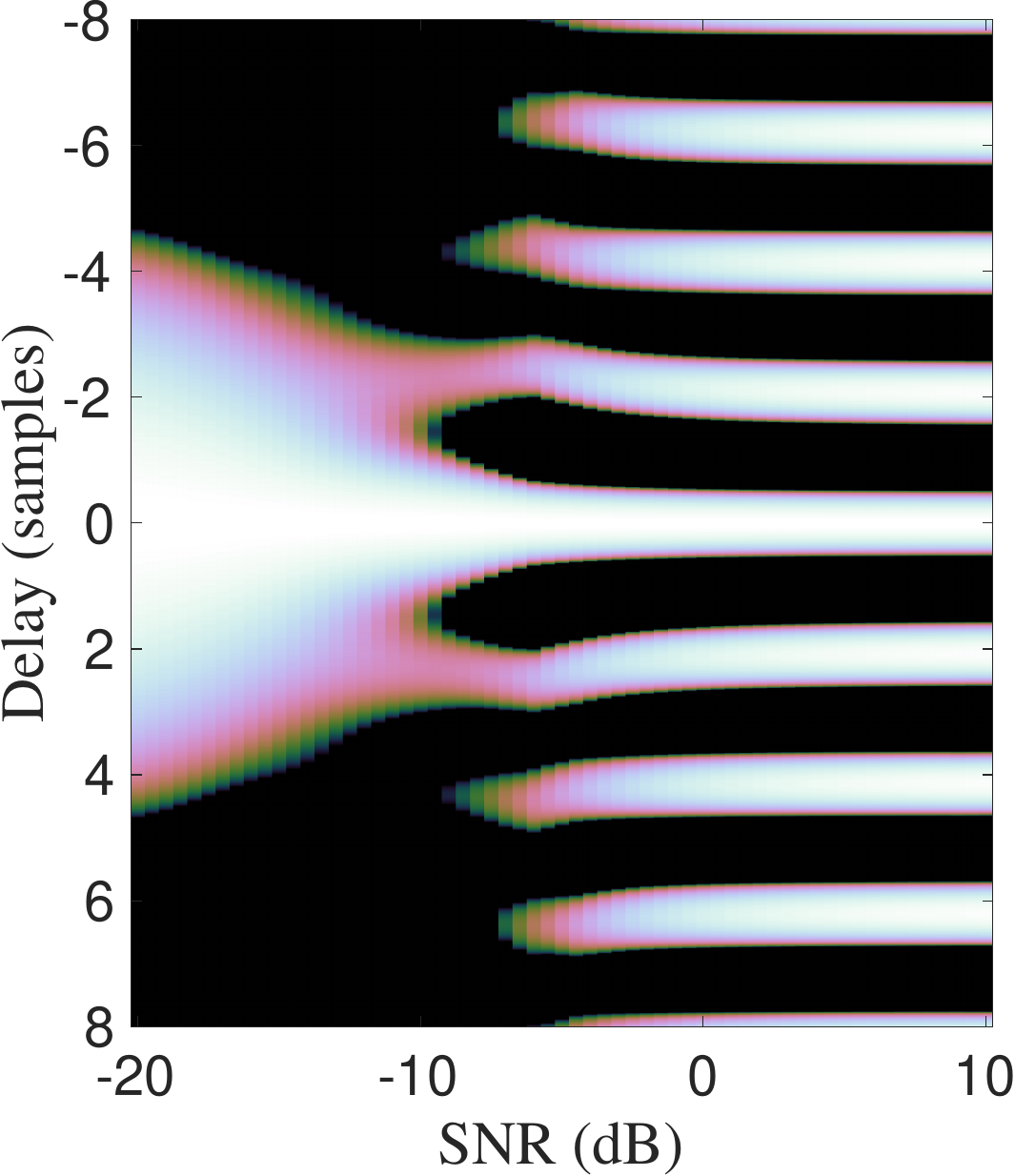}
		\caption{}
		\label{fig:acf_coh}
	\end{subfigure}
	\begin{subfigure}{0.23\linewidth}
		\centering
		\includegraphics[width=\textwidth]{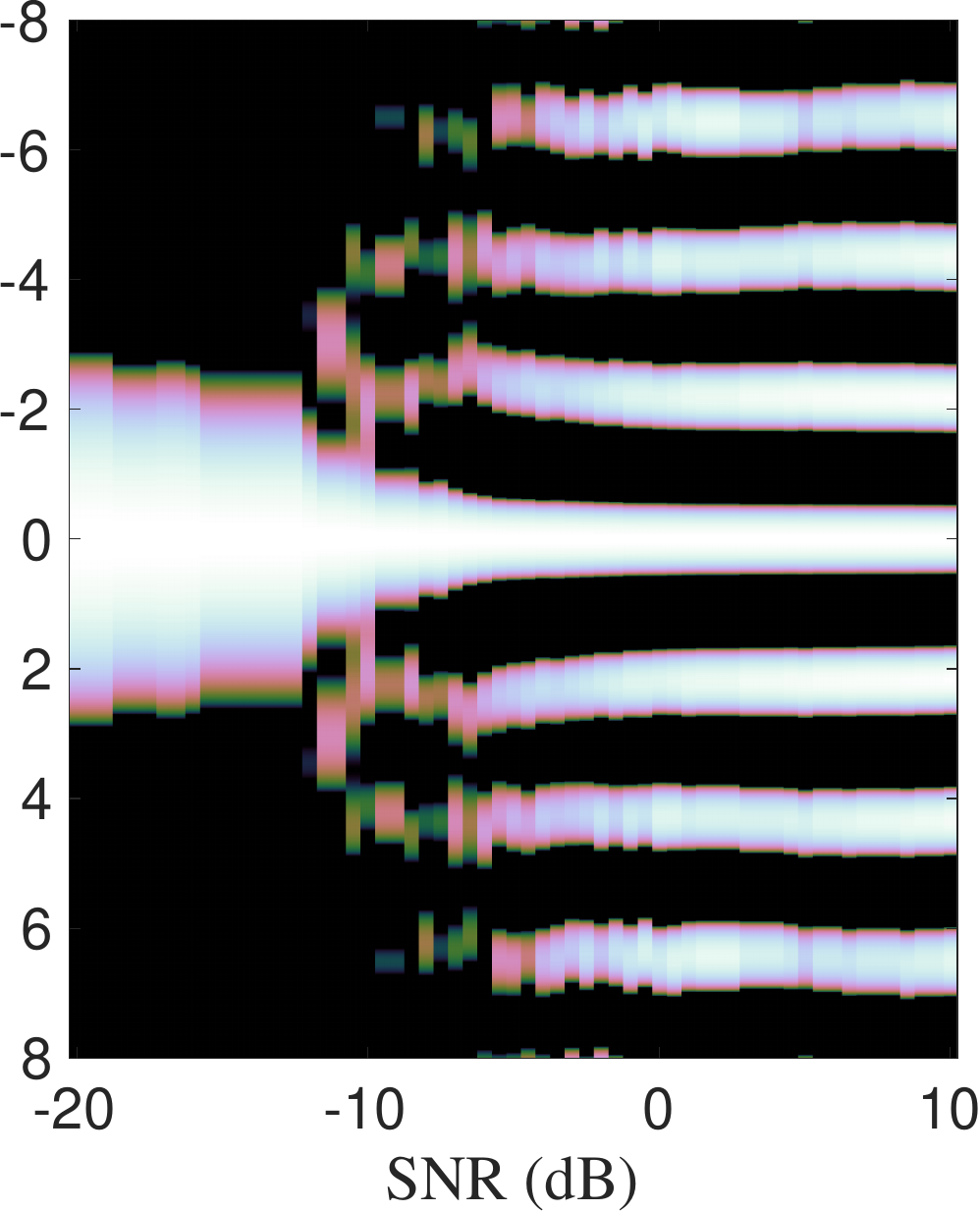}
		\caption{}
		\label{fig:acf_int_coh}
	\end{subfigure}
	\begin{subfigure}{0.23\linewidth}
		\centering
		\includegraphics[width=\textwidth]{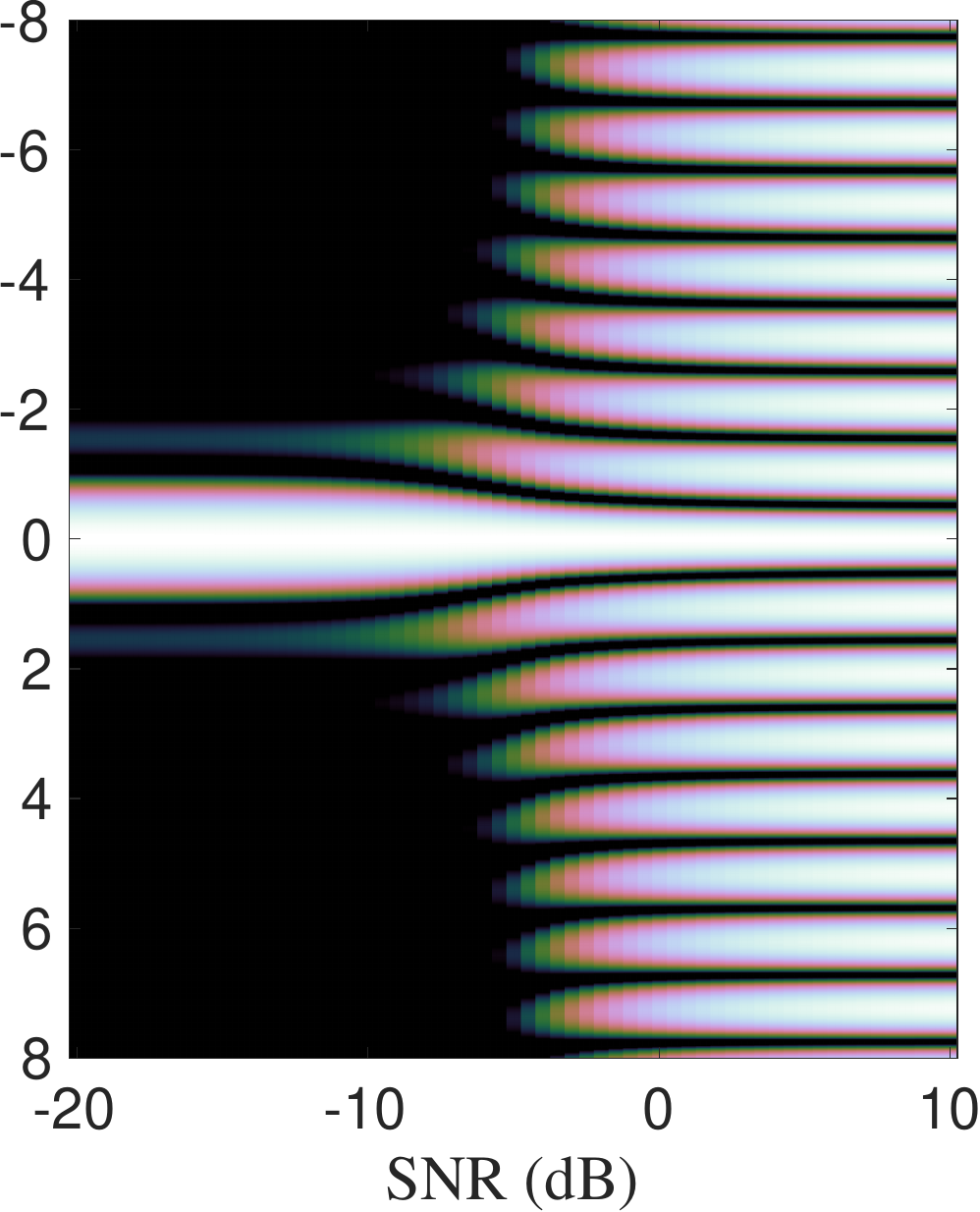}
		\caption{}
		\label{fig:acf_incoh}
	\end{subfigure}
	\begin{subfigure}{0.26\linewidth}
		\centering
		\includegraphics[width=\textwidth]{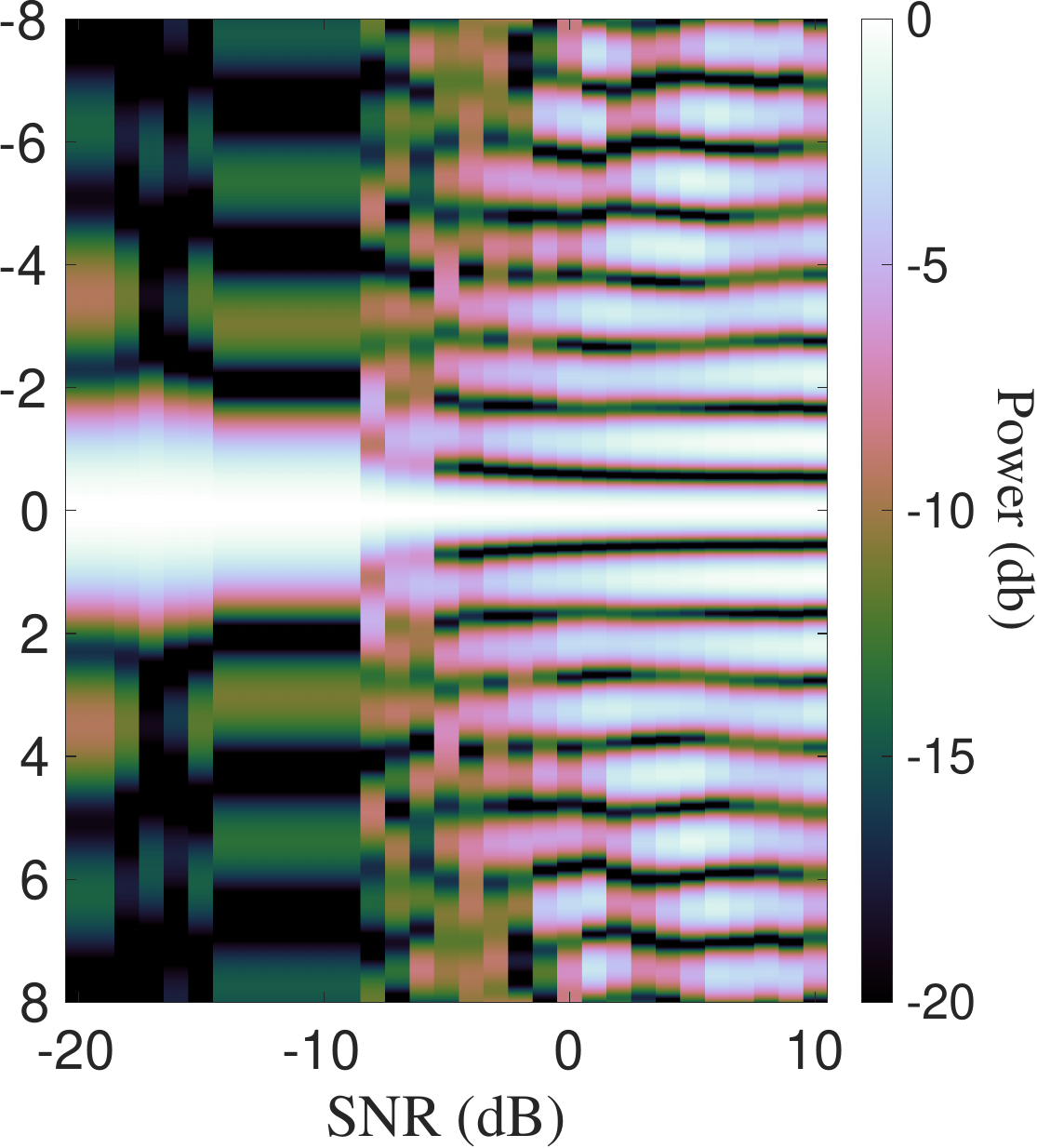}
		\caption{}
		\label{fig:acf_int_incoh}
	\end{subfigure}
	\caption{The \acp{acf} obtained from the allocations in Fig.~\ref{fig:rho_list} over a range of \acp{snr} for the coherent convex-constrained problem (a), the coherent integer-constrained problem (b), the noncoherent convex-constrained problem (c), and the noncoherent integer-constrained problem (d). Negative values in the coherent \acp{acf} are rounded up to \SI{-20}{\decibel}.}
	\label{fig:acf_list}
\end{figure*}

\begin{figure}
	\centering
	\includegraphics[width=0.9\linewidth]{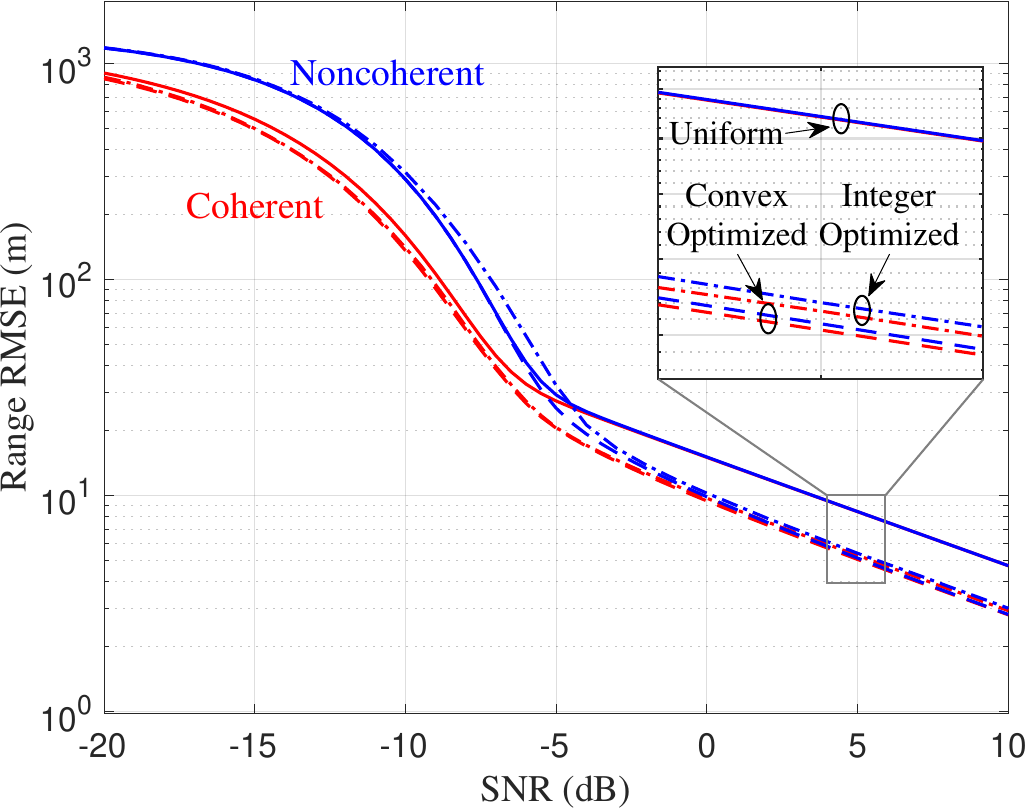}
	\caption{The \ac{zzb} for the uniform allocation, convex-constrained optimized allocation, and the integer-constrained optimized allocation. Results are shown for both the coherent and noncoherent schemes.}
	\label{fig:zzb_combined}
\end{figure}



Fig.~\ref{fig:rho_list} visualizes the resulting optimized power allocations across a range of \acp{snr}. Fig.~\ref{fig:rho_coh} shows the coherent convex-optimized allocations, Fig.~\ref{fig:rho_int_coh} shows the coherent integer-optimized allocations, Fig.~\ref{fig:rho_incoh} shows the noncoherent convex-optimized allocations, and Fig.~\ref{fig:rho_int_incoh} shows the noncoherent integer-optimized allocations. The \ac{toa} \acp{acf} obtained from these allocations are visualized in Fig.~\ref{fig:acf_list} in the same order. Fig.~\ref{fig:zzb_combined} plots the \ac{zzb} \ac{toa} \acp{rmse} for all four variations and compares these against the coherent and noncoherent \acp{zzb} obtained from a uniform power allocation across all subcarriers.

There are common trends seen in the optimal allocations and the resulting \acp{acf} over the \ac{snr} range for all four variations. At low \ac{snr}, power is allocated in a manner that creates a wide mainlobe and minimizes sidelobes within the \text{a priori} \ac{toa} region. As \ac{snr} increases, sidelobes become more tolerable, allowing power to be allocated closer to the extremities of the spectrum to sharpen the mainlobe. At high \ac{snr}, the power allocation approaches the \ac{crlb}-optimal design of allocating all power to the outermost subcarriers yet maintains an appropriate amount of power in the inner subcarriers to prevent grating lobe ambiguities in the \ac{acf}.  The \ac{zzb} is a powerfully general optimization criterion: it captures the transition from the low \ac{snr} regime to the high \ac{snr} regime.

\subsection{Discussion}

Consider the coherent and noncoherent cases in greater detail, starting with the coherent case, whose convex-optimized allocations are shown in Fig.~\ref{fig:rho_coh} and whose integer-optimized allocations are shown in Fig.~\ref{fig:rho_int_coh}. Due to the restrictions on subcarrier count and power, the integer-optimized allocations show a more gradual transition toward the \ac{crlb}-optimal choice of placing all power in the subcarrier extremities, placing power in a small number of central subcarriers to control sidelobe levels appropriately. Even with these constraints, the integer-optimized allocations create an \ac{acf}, shown in Fig.~\ref{fig:acf_int_coh}, that closely resembles the convex-optimized \ac{acf}, shown in Fig.~\ref{fig:acf_coh}. The branch-and-bound algorithm is capable of determining the best placement of these subcarriers while avoiding impractical brute-force searches. Moreover, Fig.~\ref{fig:zzb_combined} shows that the integer-optimized allocations achieve a \ac{zzb} that is negligibly worse than the \ac{zzb} achieved by the convex-optimized allocations over the entire \ac{snr} range. These are powerful results: near-optimal \ac{toa} performance can be achieved while only using a small fraction of the available spectral resources, allowing allocation of the unused subcarriers for communications or other purposes.

Next consider the noncoherent case. Fig.~\ref{fig:rho_incoh} visualizes the noncoherent convex-optimized allocations, which have a notably different distribution than the those in the coherent case shown in Fig.~\ref{fig:rho_coh}. At low \ac{snr}, the power is distributed more uniformly across the spectrum, with approximately equal spacing between subcarriers containing significant power. This spacing permits grating lobe ambiguities in the \ac{acf} that fall outside of the \textit{a priori} \ac{toa} duration. As \ac{snr} increases, the spacing shifts slightly and power is increasingly allocated to the outermost subcarriers. Fig.~\ref{fig:acf_incoh} shows the noncoherent \acp{acf} of the optimized power allocations, depicting how these power allocations transition from minimal sidelobe levels at low \ac{snr} to minimal mainlobe width and tolerable sidelobes at high \ac{snr}. In contrast to the \acp{acf} in Fig.~\ref{fig:acf_coh}, many more sidelobes are present in Fig.~\ref{fig:acf_incoh} because sidelobes that previously took negative amplitudes in the coherent \ac{acf} take positive values in the noncoherent \ac{acf}. This increases the likelihood of sidelobe-dominated \ac{toa} errors and provides insight into the differences in optimal power allocation between coherent and noncoherent cases.

Fig.~\ref{fig:rho_int_incoh} visualizes the noncoherent integer-optimized allocations, which exhibit a unique structure. The noncoherent \ac{acf} has no dependence on the absolute placement of subcarriers within the spectrum and instead depends only on relative placements of subcarriers with respect to each other. This is in stark contrast to the coherent \ac{acf} which does factor in absolute subcarrier placement. Without loss of generality, the noncoherent integer-constrained problem was initialized with power placed in subcarrier $-31$. At low \ac{snr}, the optimized power allocations take a sparse pattern which does not span the entire bandwidth. A rapid transition then occurs from \SI{-8}{\decibel} to \SI{-5}{\decibel} where the allocations expand to exploit the entire bandwidth. As \ac{snr} increases further, more power is allocated into the outermost subcarriers and fewer subcarriers are required in the center to maintain appropriate sidelobe levels. The transition toward allocating all power into the outermost subcarriers is slightly slower than that seen in the coherent integer-optimized allocations. Fig.~\ref{fig:acf_int_incoh} visualizes the noncoherent \acp{acf} of these integer-optimized allocations. Unlike the previous variations, these \acp{acf} have much higher sidelobes at low \ac{snr}, even compared against the coherent integer-optimized \acp{acf} in Fig.~\ref{fig:acf_int_coh}. This highlights the importance of using the noncoherent \ac{zzb} for optimization, as sidelobes are more difficult to control in the noncoherent case and coherent estimation may be infeasible at lower \acp{snr}.

Fig.~\ref{fig:zzb_combined} quantifies the theoretical performance of the proposed optimized allocations for \ac{toa} estimation. In the coherent case, both the convex-optimized and integer-optimized allocations outperform the uniform allocation across the entire range of \acp{snr}. Above \SI{-7}{\decibel} \ac{snr}, the optimized allocations achieve significantly reduced \ac{toa} \acp{rmse}, ultimately reducing the \ac{toa} \ac{rmse} by up to \SI{40}{\percent} in the high \ac{snr} regime. The \ac{rmse} reduction continues into the low \ac{snr} regime, albeit less significantly, where the optimized allocations are able to exploit the receiver's \textit{a priori} \ac{toa} knowledge. In the noncoherent case, the convex-optimized allocations achieve \acp{rmse} similar to the uniform allocation's \acp{rmse} in the low \ac{snr} regime but outperform the uniform allocation above \SI{-6}{\decibel} \ac{snr}. However, the noncoherent integer-optimized allocations experience an earlier thresholding effect below \SI{-4.5}{\decibel} and \ac{toa} errors increase faster than those with the uniform power allocation. This gap is explained by noting that the uniform allocation uses all $64$ subcarriers as pilot resources, whereas the integer-optimized allocations only requires that $8$ subcarriers be dedicated as pilot resources. In the high \ac{snr} regime, the noncoherent integer-optimized allocations achieve reduced \ac{toa} \acp{rmse} that come close to the \acp{rmse} achieved by the convex-optimized allocations.


\section{Experimental Results}
\label{sec:exp_results}

\begin{table}[t]
	\centering
	\caption{SDR Parameters}
	\begin{tabular}[c]{ll}
		\toprule
		\multicolumn{2}{c}{OFDM Parameters} \\
		\midrule
		$K$ & 64 subcarriers \\
		$\Delta_{\text{f}}$ (Subcarrier Spacing) & $\sfrac{3.125}{256}$\SI{}{\mega\hertz} \\
		Symbol Duration & \SI{81.92}{\micro\second} \\
		Cyclic Prefix Duration &  \SI{40.96}{\micro\second} \\
		Bandwidth & \SI{781.25}{\kilo\hertz} \\
		\toprule
		\multicolumn{2}{c}{Measurement Parameters} \\
		\midrule
		TX Gain & \SIrange{6}{33}{\decibel} \\
		RX Gain & \SI{10}{\decibel} \\
		$\fc$ (Carrier Frequency) & \SI{2.5}{\giga\hertz} \\
		Digital Offset & \SI{781.25}{\kilo\hertz} \\
		Sample Rate & \SI{3.125}{\mega{}S\per\second} \\
		$z_{0}$ (Relative Delay) & \SI{10.441}{\micro\second} \\
		$\Ta$ (\ac{toa} Prior) & \SI{20.48}{\micro\second} \\
		TX-RX Distance & $\sim$\SI{25}{\centi\meter}\\
		\bottomrule
	\end{tabular}
	\label{tab:usrp_params}
\end{table}

In addition to the theoretical results in Sec.~\ref{sec:num_results}, the \ac{toa} estimation performance of the noncoherent optimized allocations was evaluated on an \ac{sdr} measurement platform. Only the noncoherent allocations are evaluated because they do not require prior knowledge of the carrier phase, making these allocations more applicable to practical ranging systems.

\subsection{Measurement Process}

Experiments were conducted with two Ettus USRP N200s with synchronized clocks and timing. The transmit and receive antennas were spaced approximately \SI{25}{\centi\meter} apart with an unobstructed line-of-sight propagation path. The parameters of the USRP devices are listed in Table~\ref{tab:usrp_params}. The USRPs operated with direct conversion frontends, resulting in a strong DC component. To mitigate DC noise, the DC offset was calibrated prior to measurement and the \ac{ofdm} signal was digitally offset by \SI{781.25}{\kilo\hertz}. The subcarrier noise variance $\hat{\sigma}^2$ was also measured from $1000$ \ac{ofdm} symbols sampled at the receiver USRP with the transmitter USRP off. The measurement process consisted of an \ac{snr} measurement stage and a \ac{toa} measurement stage.

In the \ac{snr} measurement stage, $M_{\text{SNR}} = 250$ \ac{ofdm} symbols with a uniform pilot allocation are transmitted and received. After reception, the receiver USRP computes the complex correlation function as in (\ref{eq:complex_corr}) for all symbols, denoted for the $m$th symbol as $\tilde{A}_{m}\left(z,\bm{x}\right)$ for $m = 0, 1, \ldots, M_{\text{SNR}}-1$. The receiver then estimates the integrated \ac{snr} from the peak power of the coherent sum of all complex correlation functions
\begin{align}
	\hat{\gamma} = \max_{z \in [0,\Na]} \frac{|\sum_{m=0}^{M_{\text{SNR}}}\tilde{A}_{m}\left(z,\bm{x}\right)|^2}{\hat{\sigma}^2}.
\end{align}
The \ac{snr} can then be estimated as $\hat{\gamma}/K$.

In the \ac{toa} measurement stage, the transmitter selects the \ac{ofdm} allocation optimized for an \ac{snr} closest to the estimated \ac{snr} and then transmits $M_{\text{TOA}} = \SI{25e3}{}$ symbols generated from the chosen allocation. Similar to the \ac{snr} measurement stage, the receiver USRP then computes the complex correlation function, denoted for symbol $m$ as $\tilde{A}_{m}\left(z,\bm{x}\right)$ for $m = 0, 1, \ldots, M_{\text{TOA}}-1$. The receiver then estimates the \ac{toa} using the noncoherent estimator in (\ref{eq:mle_incoh}), which is expressed in units of samples as
\begin{align}
	\hat{z}_{m} = \max_{z \in [0,\Na]} |\tilde{A}_{m}\left(z,\bm{x}\right)|^2.
\end{align}
Prior to estimating the \ac{rmse}, the receiver discards the first $M_{\text{init}} = 50$ \ac{toa} estimates to eliminate any errors that could be caused by samples collected before the USRPs' RF components have settled to their tuning configuration. These measurements have a mean
\begin{align}
	\bar{z} = \frac{1}{M_{\text{TOA}} - M_{\text{init}}}\sum_{m=M_{\text{init}}}^{M_{\text{TOA}}-1} \hat{z}_m
\end{align}
and an unbiased \ac{rmse}
\begin{align}
	\hat{\sigma}_{z} = \sqrt{\frac{\sum_{m=M_{\text{init}}}^{M_{\text{TOA}}-1} (\hat{z}_m - \bar{z})^2 }{M_{\text{TOA}}-M_{\text{init}}-1}}.
\end{align}
This \ac{rmse} can be expressed in units of seconds by scaling $\hat{\sigma}_{z}$ by $\Ts$ or expressed in units of meters by scaling $\hat{\sigma}_{z}$ by $c\Ts$, where $c$ is the speed of light in \SI{}{\meter\per\second}. This measurement process was repeated for transmit gains from \SI{6}{\decibel} to \SI{33}{\decibel} in intervals of \SI{3}{\decibel}.

\subsection{Measurement Results}

\begin{figure}[h!]
	\centering
	\includegraphics[width=0.9\linewidth]{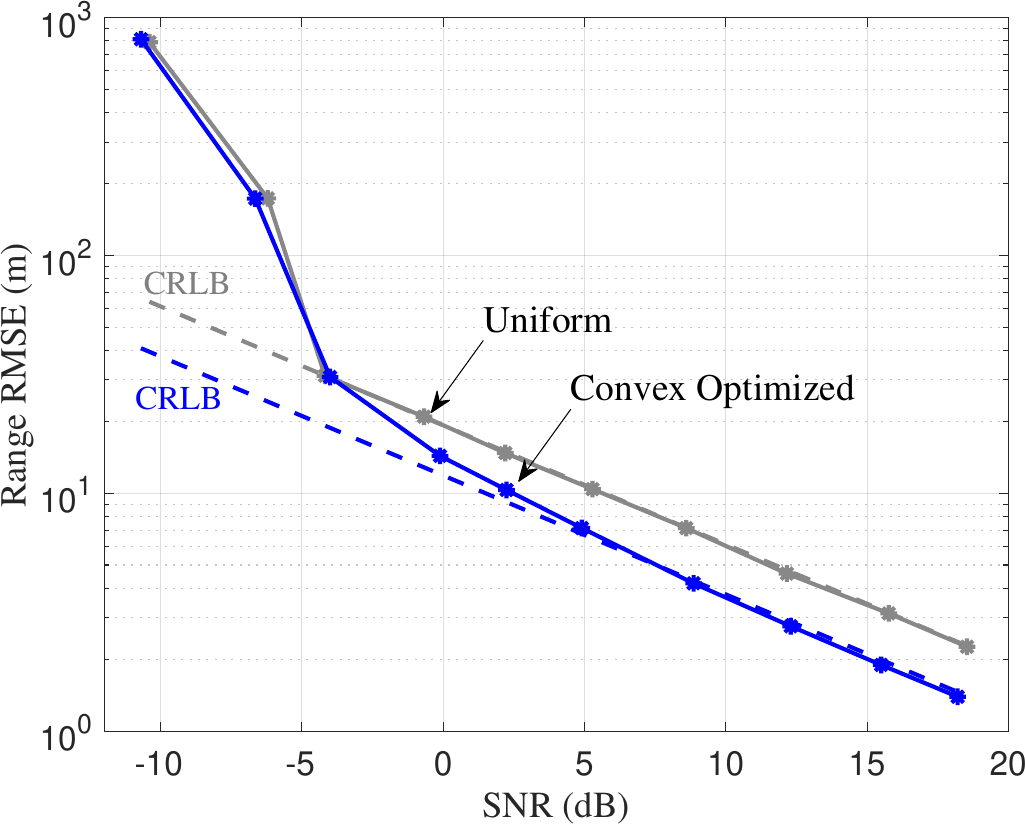}
	\caption{Plots of the \ac{toa} \acp{rmse} measured on the \ac{sdr} platform using both the noncoherent optimized power allocation and a uniform power allocation. Results are compared against the \ac{crlb}.}
	\label{fig:sdr_incoh}
\end{figure}

\begin{figure}[h!]
	\centering
	\includegraphics[width=0.9\linewidth]{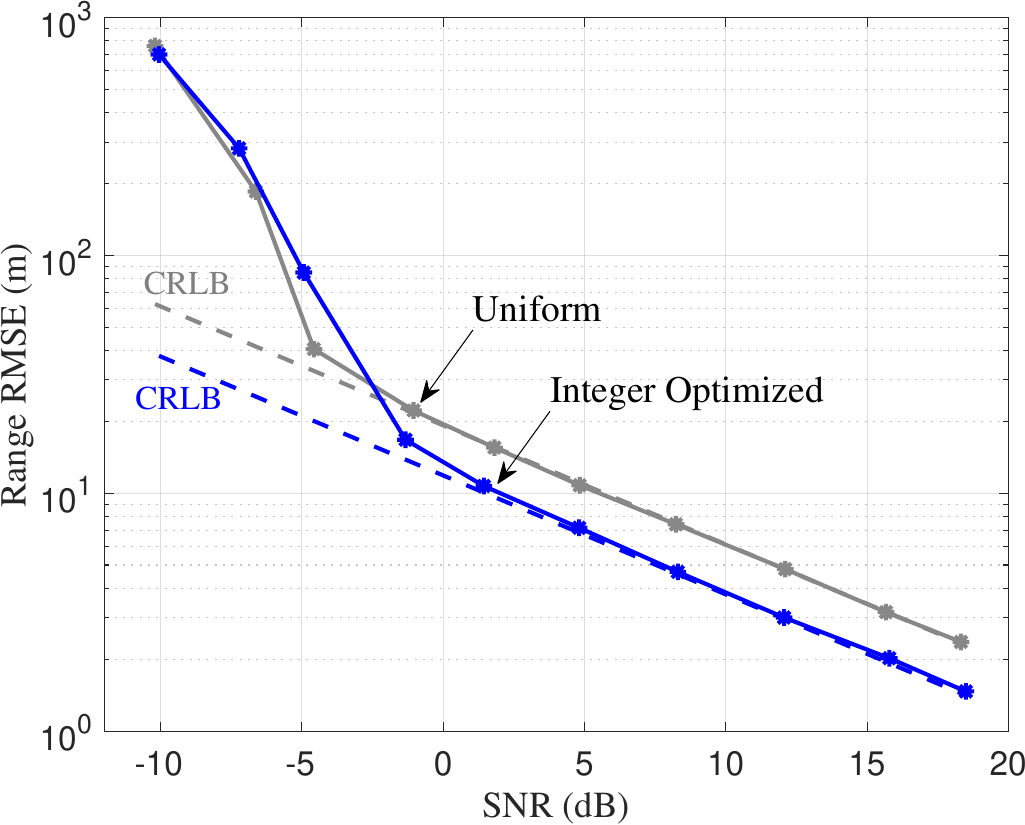}
	\caption{Plots of the \ac{toa} \acp{rmse} measured on the \ac{sdr} platform using both the noncoherent integer-optimized power allocation and a uniform power allocation. Results are compared against the \ac{crlb}.}
	\label{fig:sdr_incoh_bin}
\end{figure}

Fig.~\ref{fig:sdr_incoh} plots the measured \acp{rmse} against the estimated \acp{snr} collected on the \ac{sdr} platform using both a uniform allocation and the noncoherent convex-optimized allocation. As in the theoretical results, the noncoherent convex-optimized allocation achieves an \ac{rmse} much better than the uniform allocation for moderate to high \acp{snr}. Above \SI{1}{\decibel} \ac{snr}, the measured \acp{rmse} for both allocations closely follow their \acp{crlb}.

Fig.~\ref{fig:sdr_incoh_bin} plots the measured \acp{rmse} against the estimated \acp{snr} collected on the \ac{sdr} platform using both a uniform allocation and the noncoherent integer-optimized allocation. Compared to the convex-optimized allocation, the measured \ac{rmse} from the integer-optimized allocation diverges from the \ac{crlb} at a higher \ac{snr} and has greater errors than the measured \ac{rmse} from the uniform allocation below approximately \SI{-3}{\decibel} \ac{snr}. This crossing point occurs at a slightly higher \ac{snr} than in the theoretical results, which is expected since the propagation channel is not an ideal \ac{awgn} channel in practice. Despite its worse performance at low \ac{snr}, the integer-optimized allocation may be a desirable allocation since it only requires $8$ pilot subcarriers, allowing data or other resources to be multiplexed in frequency unlike the uniform allocation which requires all $64$ subcarriers.


\section{Conclusions}
\label{sec:conclusion}
This paper has demonstrated how \ac{ofdm} pilot allocations can be optimized to minimize \ac{toa} estimation errors in a manner that accounts for the low \ac{snr} thresholding effects caused by the presence of sidelobes in the signal's  \ac{acf}. This optimization was conducted by minimizing the \ac{zzb} and was analyzed under both coherent and noncoherent reception. This paper proved the convexity of the problem, provided readily-usable expressions for the gradients, and solved for optimal allocations over a wide range of \acp{snr}. The theoretical \ac{toa} error variances achieved by the optimal allocations were compared against the error variances achieved by a uniform allocation. Integer constraints were then introduced into the optimization problem, in which the pilot allocations are restricted to a sparse selection of subcarriers with equally-distributed power. A branch-and-bound algorithm was proposed for solving this integer-constrained problem to achieve near-optimal allocations that only require a sparse subset of subcarriers, solving for these allocations over a wide range of \acp{snr} and comparing the theoretical \ac{toa} error variances against uniform allocations. The real-world applicability of these pilot allocations was demonstrated by measuring the \ac{toa} errors obtained from \acp{sdr} transmitting and receiving the optimized signals.

These results illustrate how intelligent allocation of pilot resources within \ac{ofdm} signals can yield notable improvements in \ac{toa} estimation, even when constrained to using a sparse subset of subcarriers. Such optimal allocations will be a key enabler of dual-functional and \ac{isac} systems that aim to achieve precise \ac{toa}-based user positioning with a minimal number of resources allocated for positioning, freeing up resources for communications.

\section*{Acknowledgments}

This work was supported by the U.S. Space Force under an STTR contract with
Coherent Technical Services, Inc., and by affiliates of the 6G@UT center within
the Wireless Networking and Communications Group at The University of Texas at
Austin.

\appendices

\section{Coherent Error Probability Gradient and Hessian}
\label{sec:app_coh_grad_hess}
First, the gradients of the coherent probability of error $\PminCR$ with respect to the mapped power allocations $\tilde{\bm{\rho}}$, expressed in (\ref{eq:grad_pmin_coh}), will be derived. Taking the partial derivatives of (\ref{eq:Pmin_coh}) yields
\begin{align}
	&\frac{\partial}{\partial \tilde{\rho}_n} \PminCR \nonumber\\
	&= \frac{\partial}{\partial \tilde{\rho}_n} \frac{1}{2} - \frac{1}{2} \text{erf}\left(\sqrt{ \frac{\gamma}{2} \left( 1 - \ACcR \right) } \right) \nonumber\\
	&= \frac{\partial}{\partial \tilde{\rho}_n} \frac{1}{2} - \frac{1}{\sqrt{\pi}} \int_{0}^{\sqrt{ \frac{\gamma}{2} \left( 1 - \ACcR \right) }} \exp(-t^2) dt \nonumber\\
	&= \frac{-1}{\sqrt{\pi}} \exp\left( \frac{-\gamma}{2} \left( 1 {-} \ACcR \right) \right) \nonumber\\
	&\quad\quad\times\left( \frac{\partial}{\partial \tilde{\rho}_n}  \sqrt{ \frac{\gamma}{2} \left( 1 {-} \ACcR \right) } \right) \nonumber\\
	&= \frac{-1}{2\sqrt{\pi}} \exp\left( \frac{-\gamma}{2} \left( 1 {-} \ACcR \right) \right) \frac{\frac{\partial}{\partial \tilde{\rho}_n} \frac{-\gamma}{2}\ACcR}{\sqrt{ \frac{\gamma}{2} \left( 1 {-} \ACcR \right) } } \nonumber\\
	&= \frac{\sqrt{\gamma}}{2\sqrt{2\pi}} \exp\left( \frac{-\gamma}{2} \left( 1 {-} \ACcR \right) \right) \frac{\cos(2\pi z d[n{+}1]/K) {-} 1}{\sqrt{ 1 {-} \ACcR } }
	\label{eq:grad_coh_deriv}
\end{align}

Next, the Hessian of $\PminCR$ with respect to $\tilde{\bm{\rho}}$ is derived. Taking the partial derivatives of (\ref{eq:grad_coh_deriv}) yields
\begin{align}
	&\frac{\partial^2}{\partial \tilde{\rho}_n \partial \tilde{\rho}_m} \PminCR \nonumber\\
	&= \frac{\partial}{\partial \tilde{\rho}_m} \overbrace{\frac{\sqrt{\gamma}}{2\sqrt{2\pi}} \exp\left( \frac{-\gamma}{2} \left( 1 - \ACcR \right) \right)}^{A_{2}} \nonumber\\
	&\quad\times \underbrace{\frac{\cos(2\pi z d[n{+}1]/K) - 1}{\sqrt{ 1 - \ACcR } } }_{B_{2}} \nonumber\\
	&= \left(\frac{\partial}{\partial \tilde{\rho}_m} A_{2}\right) B_{2} + A_{2} \left(\frac{\partial}{\partial \tilde{\rho}_m} B_{2}\right).
	\label{eq:hess_coh_deriv}
\end{align}
The partial derivatives of $A_{2}$ and $B_{2}$ are
\begin{align}
	\frac{\partial}{\partial \tilde{\rho}_m} A_{2}
	&= \frac{\gamma^{3/2}}{4\sqrt{2\pi}} \exp\left( \frac{-\gamma}{2} \left( 1 - \ACcR \right) \right) \nonumber\\
	&\times \left(\cos(2\pi z d[m{+}1]/K) - 1\right),
\end{align}
and
\begin{align}
	&\frac{\partial}{\partial \tilde{\rho}_m} B_{2} \nonumber\\
	&= \frac{-(\cos(2\pi z d[n{+}1]/K) - 1)(\cos(2\pi z d[m{+}1]/K) - 1)}{ 2\left( 1 - \ACcR \right)^{3/2} }.
\end{align}
Through some algebra, (\ref{eq:hess_coh_deriv}) becomes
\begin{align}
	&\frac{\partial^2}{\partial \tilde{\rho}_n \partial \tilde{\rho}_m} \PminCR \nonumber\\
	&= \frac{\sqrt{\gamma}}{4\sqrt{2\pi}} \exp\left( \frac{-\gamma}{2} \left( 1 - \ACcR \right) \right) \nonumber\\
	&\times \frac{\left(\cos(2\pi z d[n{+}1]/K) - 1\right)\left(\cos(2\pi z d[m{+}1]/K) - 1\right)}{\sqrt{1-\ACcR}} \nonumber\\
	&\times \left( \gamma + \left(1-\ACcR\right)^{-1} \right).
\end{align}

\section{Noncoherent Error Probability Gradient}
\label{sec:app_noncoh_grad}
The gradients of the noncoherent probability of error $\PminNR$ with respect to the mapped power allocations $\tilde{\bm{\rho}}$, expressed in (\ref{eq:grad_pmin_incoh}), are derived.
Through application of the multivariable chain rule, the partial derivatives of the Marcum-Q term are
\begin{align}
	\frac{\partial}{\partial \tilde{\rho}_n} Q_{1}(a,b) = &\left(\frac{\partial}{\partial \tilde{\rho}_n} a\right)\left(\frac{\partial}{\partial a} Q_{1}(a,b)\right) \nonumber\\
	+ &\left(\frac{\partial}{\partial \tilde{\rho}_n} b\right)\left(\frac{\partial}{\partial b} Q_{1}(a,b)\right).
	\label{eq:app_noncoh_q}
\end{align}
The partial derivatives of the Marcum-Q with respect to its inputs $a$ and $b$ are \cite{annamalai2008simple}
\begin{align}
	\frac{\partial}{\partial a} Q_{1}(a,b) &= b \exp{\left(-\gamma/2\right)} I_{1}(ab),\\
	\frac{\partial}{\partial b} Q_{1}(a,b) &= -b \exp{\left(-\gamma/2\right)} I_{0}(ab),
\end{align}
while the partial derivatives of $a$ and $b$ with respect to $\tilde{\bm{\rho}}$ are
\begin{align}
	\frac{\partial}{\partial \tilde{\rho}_n} a &=  \frac{\frac{\sqrt{\gamma}}{4\sqrt{2}}\left(\frac{\partial}{\partial \tilde{\rho}_n} \ACnR\right)  }{\sqrt{1-\ACnR}\sqrt{1-\sqrt{1-\ACnR}}},\\
	\frac{\partial}{\partial \tilde{\rho}_n} b &= \frac{\frac{-\sqrt{\gamma}}{4\sqrt{2}}\left(\frac{\partial}{\partial \tilde{\rho}_n} \ACnR\right)  }{\sqrt{1-\ACnR}\sqrt{1+\sqrt{1-\ACnR}}}.
	\label{eq:app_grad_zzb_incoh}
\end{align}
The partial derivatives of the noncoherent  autocorrelation function with respect to $\tilde{\bm{\rho}}$ are
\begin{align}
	&\frac{\partial}{\partial \tilde{\rho}_n} \ACnR \nonumber\\
	&= \sum_{k=0}^{K-2} 2\tilde{\rho}[k] (\cos(2 \pi z d[k+1] / K)-1) \cos(2 \pi z d[n{+}1] / K) \nonumber\\
	&+ \sum_{k=0}^{K-2} 2\tilde{\rho}[k] \sin(2 \pi z d[k+1] / K) \sin(2 \pi z d[n{+}1] / K).
	\label{eq:app_noncoh_acf}
\end{align}
Simple substitution of (\ref{eq:app_noncoh_q}-\ref{eq:app_noncoh_acf}) yields the gradients of the Marcum-
Q term in (\ref{eq:grad_pmin_incoh}). Next, the partial derivatives of the modified Bessel function term in (\ref{eq:grad_pmin_incoh}) are
\begin{align}
	&\frac{\partial}{\partial \tilde{\rho}_n} \frac{-1}{2}\exp{\left(\sfrac{-\gamma}{2}\right)}I_{0}(ab) \nonumber\\
	&= \frac{\partial}{\partial \tilde{\rho}_n} \frac{-1}{2}\exp{\left(\sfrac{-\gamma}{2}\right)}I_{0}(\sfrac{\gamma}{2}\sqrt{\ACnR}) \nonumber\\
	&= \frac{-1}{2}\exp{\left(\sfrac{-\gamma}{2}\right)}I_{1}(\sfrac{\gamma}{2}\sqrt{\ACnR})\frac{\gamma\left(\frac{\partial}{\partial \tilde{\rho}_n} \ACnR\right)}{4\sqrt{\ACnR}}.
	\label{eq:app_noncoh_bessel}
\end{align}
Finally, the expressions for (\ref{eq:app_noncoh_q}) and (\ref{eq:app_noncoh_bessel}) can be substituted into (\ref{eq:grad_pmin_incoh}) to obtain the desired gradient of $\PminNR$.

\bibliographystyle{IEEEtran} 
\bibliography{pangea}
\end{document}